\documentclass[11pt]{article}

\usepackage{geometry}
\usepackage{thmtools,thm-restate} 
\usepackage{bbm}
\usepackage{amsthm,amsmath,amssymb}
\usepackage{graphicx}
\usepackage{enumerate}
\usepackage[dvipsnames]{xcolor}
\usepackage{url}
\usepackage[ruled]{algorithm2e} % For algorithms

\SetAlFnt{\small}
\SetAlCapFnt{\small}
	\SetAlCapNameFnt{\small}
\SetAlCapHSkip{0pt}
\IncMargin{-1.6\parindent}
\usepackage[colorlinks=true,citecolor=blue,linkcolor=blue,urlcolor=blue]{hyperref}

\theoremstyle{plain}
\newtheorem{theorem}{Theorem}
\newtheorem{lemma}[theorem]{Lemma}

\newtheorem{proposition}[theorem]{Proposition}

\theoremstyle{definition}

\theoremstyle{remark}

\usepackage{tikz}
\usetikzlibrary{arrows,automata,positioning}
\usepackage[latin1]{inputenc}
\usepackage{pgfplots}
\pgfplotsset{compat=1.13}

\usetikzlibrary{math} %needed tikz library
\tikzmath{
\m = 1; \var = 0.82*\m; \s= (0.82*\m)^0.5;
\L = 2.5*\m; \c = \m/\s;
\pleft = 1/(1+\c^2);
\Tpleft = \m + (\pleft*(1-\pleft))^0.5*\s;
\pright = 1/(1 + \s^2/(\L-\m)^2);
\Tpright = \m + (\pright*(1-\pright))^0.5*\s;
\Ttwopleft = \m + (\Tpright - \m)*\pleft + (\pleft*(1-\pleft))^0.5*\s;
\Ttwopright = \m + (\Tpright - \m)*\pright + (\pright*(1-\pright))^0.5*\s;
\Ttwopone = \m + (\Tpright - \m);
\sji = 0.5;
\mtwo = 0.4;
\Ltwo = 1;
\mthree = 0.8;
\sjithree = 0.55;
\mfour = 0.3;
}

\usepackage{comment}
\usepackage{relsize}

\newcommand{\dd}{\text{d}}
\newcommand{\E}{\mathbb{E}}
\newcommand{\N}{\mathbb{N}}
\newcommand{\R}{\mathbb{R}}

\newcommand{\Prob}{\mathbb{P}}

\newcommand{\M}{\mathbb{M}}
\newcommand{\G}{\mathbb{G}}

\usepackage{ifthen}
% ...
\numberwithin{theorem}{section}
\numberwithin{equation}{section}
\begin{document}
\title{
Optimal Stopping Theory for a Distributionally Robust Seller 
}
\author{
Pieter Kleer and Johan S.H. van Leeuwaarden\\\
Tilburg University\\
Department of Econometrics and Operations Research\\
%Tilburg, The Netherlands\\
\texttt{\{p.s.kleer,j.s.h.vanleeuwaarden\}@tilburguniversity.edu}
}

	\maketitle

\begin{abstract}
Sellers in online markets face the challenge of determining the right time to sell in view of uncertain future offers. Classical stopping theory assumes that sellers have full knowledge of the value distributions, and leverage this knowledge to determine stopping rules that maximize expected welfare. In practice, however, stopping rules must often be determined under partial information, based on scarce data or expert predictions. Consider a seller that has one item for sale and receives successive offers drawn from some value distributions.  The decision on whether or not to accept an offer is irrevocable, and the value distributions are only partially known. We therefore let the seller adopt a robust maximin strategy, assuming that value distributions are chosen adversarially by nature to minimize the value of the accepted offer. 
We provide a general maximin solution to this stopping problem that identifies the optimal (threshold-based) stopping rule for the seller for all possible statistical information structures. 
We then perform a detailed analysis for various ambiguity sets relying on knowledge about the common mean, dispersion (variance or mean absolute deviation) and support of the distributions. We show for these information structures that the seller's stopping rule consists of decreasing thresholds converging to the common mean, and that nature's adversarial response, in the long run, is to always create an all-or-nothing scenario. The maximin solutions also reveal what happens as dispersion or the number of offers grows large.
\end{abstract}

		\section{Introduction}
	\label{sec:intro}
	
A major challenge faced by sellers in online markets is to determine the right time to sell in view of uncertain future offers. We consider the classical setting of a seller that has one item for sale and receives sequentially $n$ offers.  This setting is studied since the 1960s \cite{ferguson2006optimal} and gained renewed relevance due to the rise of online algorithms.
The seller decides to sell or not, by comparing the current offer to possibly higher future offers. Rejecting the current offer holds the promise of receiving a higher offer but comes with the risk of receiving lower future offers only. 
To provide the seller a guideline for when to accept an offer, some knowledge about future offers is required. Therefore, in the area of (Bayesian) optimal stopping theory, the distributions of future offers are assumed to be known, and leveraged to obtain stopping rules for the seller.

Settings with such full distributional knowledge, however, are often criticized for being unrealistic.  What is nowadays known as ``Wilson's doctrine" in economics and game theory \cite{Wilson1987} (see also, e.g., \cite[Section 2]{Milgrom2004}) calls for reducing 
the common knowledge that is taken as given, so that the resulting analysis becomes more robust and valuable in realistic problems where only limited information is available. Inspired by this, we study a robust version of the stopping problem that assumes 
the seller has very limited distributional knowledge, for example in the form of moment information or maximal values. Optimizing a seller's welfare (or revenue) under limited information is a rich area of research  in the fields of theoretical economics, computer science and operations research, dating back to the work of Scarf \cite{scarf1958min} and falling within the scope of distributionally robust optimization. We elaborate on this literature, and other robust approaches to stopping problems, in Section \ref{sec:relwork}. 

We continue with a formal description of the  Bayesian stopping problem, and then explain in more detail our robust version with limited information. Consider the stopping problem in which a seller receives for an item a sequence of $n$ offers. The seller gets to see the offers one by one, in the order $(1,\dots,n)$,  and upon seeing an offer, has to decide irrevocably whether or not to accept (a rejected offer cannot still be accepted at a later point in time). The seller can accept at most one offer, after which the selling process stops. The goal of the seller is to come up with a stopping rule that maximizes the expected payoff of the selected offer (i.e., we consider the  problem of so-called online welfare maximization). 

It is well-known that an optimal stopping rule for the seller is to use a  threshold-based strategy, that can be computed using dynamic programming (or backwards induction), and which we explain next. Let $v_i \geq 0$ denote the (initially unknown) value of offer $i$. The seller only knows the distribution of an offer, so $v_i$ should be considered a realization of a random variable $X_i \sim \Prob_i$, where  $\Prob_i$ is a nonnegative probability distribution (or measure), for $i = 1,\dots,n$. The seller knows $\Prob_1,\dots,\Prob_n$, but not the actual values $v_i$. We further assume that the random variables $X_1,\dots,X_n$ are pairwise independent. Based on the known distributions $\Prob_1,\dots,\Prob_n$, the seller computes a threshold $T(i)$ for every offer $i$, and selects the first offer for which $v_i \geq T(i)$. This strategy yields the maximum expected payoff (i.e., the value of the selected offer), with the expectation taken with respect to the random variables $X_1,\dots,X_n$. How does one compute the thresholds? First note that once the seller reaches the last offer $n$, it is always optimal to accept it, because all values are nonnegative. 
Knowing that $T(n) = 0$, what is the optimal stopping rule for selecting offer $n-1$? If the seller does not select offer $n-1$, the expected payoff from the last offer is $\E[X_n]$. Therefore, the optimal policy is to select offer $n-1$ if and only if $v_{n-1} \geq \E[X_n]$. Hence, we set the threshold $T(n-1) = \E[X_n]$. Generalizing this line of thought, the threshold $T(i)$ is set to be the expected payoff would the seller not select offers $k = 1,\dots,i$, but instead continue. These thresholds can be computed recursively, in a backwards fashion, by
\begin{align}
T(i) & = \E[\max\{T(i+1),X_{i+1}\}] = \int_0^\infty \max\{T(i+1),x\}\, \dd\Prob_{i+1}(x). 
\label{eq:backward}
\end{align}

In our robust version of the problem, the seller does not know the full distributions $\Prob_1,\dots,\Prob_n$, but only that $\Prob_i$ belongs to some ambiguity set $\mathcal{P}(I_i)$ containing all distributions that satisfy the limited information $I_i$. For instance, if $I_i= \{\E[X_i]\}$, the ambiguity set $\mathcal{P}(I_i)$ contains all nonnegative distributions with mean $\E[X_i]$. Distributional information, like moments or maximal values, can typically be learned or estimated based on historical data.
The information structure  $\{I_1,\ldots,I_n\}$, however, is usually not sufficient for computing the thresholds in \eqref{eq:backward}, and instead the seller will have to play a maximin game against nature.  
In this game, the seller first decides on a stopping rule $\tau$ for which offer to select, after which nature gets to choose from the ambiguity sets the worst-case distributions $\{\Prob_1,\dots,\Prob_n\}$  that under the stopping rule $\tau$ minimize the expected payoff of the seller. Maximizing the (robust) expected payoff of the seller then corresponds to solving the maximin problem
\begin{align}
r^* = \max_{\tau} \min_{\forall i:\Prob_i\in \mathcal{P}(I_i)} \mathbb{E}_{\Prob_1,\dots,\Prob_n}[X_{\tau}] ,
\label{eq:minimax_seller}
\end{align}
where $X_{\tau}$ is the random variable denoting the value of the offer that was selected under stopping rule $\tau$. 
{The optimal stopping rule $\tau^*$ solving \eqref{eq:minimax_seller} is given in Theorem \ref{thm:opt_robust_thresholds}, and can be seen as a robust counterpart for the threshold-based strategy as described by the thresholds in \eqref{eq:backward}. Using an induction-based argument, inspired by the proof of a robust Bellman  equation of Iyengar \cite{iyengar2005robust} for Markov decision processes, we prove the following result:

	\begin{theorem}[Optimal robust thresholds]
The following threshold-based stopping rule $\tau^*$ solves the minimax problem \eqref{eq:minimax_seller} of the seller within the class of randomized stopping rules (formal definition is given in Section \ref{sec:stopping_rules}): Select $v_i$ if and only if $v_i \geq T(i)$ where $T(i)$ is recursively defined by
\begin{align}
T(i) & = \min_{\Prob_{i+1} \in \mathcal{P}(I_{i+1})}\E_{\Prob_{i+1}}[\max\{T(i+1),X_{i+1}\}] \nonumber\\
&=  \min_{\Prob_{i+1} \in \mathcal{P}(I_{i+1})} \int_0^{\infty} \max\{T(i+1),x\}\, \mathrm{d} \Prob_{i+1}(x) 
\label{eq:opt_robust_tresholds}
\end{align}
for $i = 0,\dots,n-1$ with $T(n) = 0$. The robust expected payoff $r^*$ of the seller is $T(0)$. In case the minimum in the definition of $T(i)$ does not exist, one should replace it by the infimum (which we will assume to exist).
\label{thm:opt_robust_thresholds}
\end{theorem}
An equivalent formulation of \eqref{eq:opt_robust_tresholds}, that we will often use, is given by 
\begin{align}
T(i) &= T(i+1) + \E_{\Prob_{i+1}}[X_i]  - \max_{\Prob_{i+1} \in \mathcal{P}(I_{i+1})}  \int_0^{\infty} \min\{T(i+1),x\}\, \mathrm{d} \Prob_{i+1}(x).
\label{eq:opt_robust_tresholds_v2}
\end{align}
{Although Theorem~\ref{thm:opt_robust_thresholds} might seem intuitive, its proof requires careful reasoning. We present the formal proof, including a formal definition of stopping rules, in Section \ref{sec:robust_optimal_thresholds}, and just sketch the main ideas here. The proof uses induction and consists of two main steps: First we argue that, based on the induction hypothesis, it is always optimal for the seller to use the thresholds in \eqref{eq:opt_robust_tresholds} for the offers $i = 2,\dots,n$. After that, we argue that it is then also optimal to use \eqref{eq:opt_robust_tresholds} for the first offer.  All this requires a lot of care because of the minimum involved in the definition of $T(i)$. }

Theorem~\ref{thm:opt_robust_thresholds} is rather general, as we only assume independence of the offer values. The optimal robust strategy defined by the recursive relation \eqref{eq:opt_robust_tresholds} applies to all forms of ambiguity of the offer value distributions. {As explained below, we will mostly apply Theorem~\ref{thm:opt_robust_thresholds} for ambiguity sets that capture the first two moments,  but one could also include higher moments, or additional information such as skewness, unimodality and tail probabilities. }But more generally, the robust framework works for all forms of ambiguity, and  also allows to condition on proximity to some reference distribution via statistical distance measures such as $\phi$-divergence and Wasserstein distance. This more broader application of Theorem~\ref{thm:opt_robust_thresholds} we leave for future research.

	\subsection{Solving the robust recursion}
	\label{sec:results}
Using Theorem~\ref{thm:opt_robust_thresholds} to derive closed-form stopping rules in specific partial information setting is the next theme of this paper. Notice that this can only work if the maximin problem in Theorem \ref{thm:opt_robust_thresholds} can be solved.
To do so, one needs to find the worst-case distribution in the ambiguity set that gives a tight bound on $\E[\max\{T(i+1), X\}]$. 
%Despite the generality of Theorem~\ref{thm:opt_robust_thresholds}, 
We shall focus on ambiguity sets that fix the mean and dispersion of the offer value distribution, where dispersion is measured in terms of variance or mean absolute devation (MAD). These intuitive summary statistics are easy to estimate when data is available, but without data can also be tuned by experts to assess risk-return tradeoffs.  Robust analysis with ambiguity sets based on summary statistics such as the mean $\mu$ and variance $\sigma^2$ dates back to the work of Scarf \cite{scarf1958min} on the (robust) newsvendor problem. Scarf showed that the worst-case distributions are two-point distributions, with positive probability mass on two points, and obtained the now famous robust optimal order quantity that solves the maximin newsvendor problem. {Scarf's paper is widely considered to be the first paper on distributionally robust optimization. The fact that two-point distributions came out as worst-case distributions can be understood through the deep connections with moment problems and semi-infinite linear programs, see Section~\ref{sec:mvs}. Such worst-case distributions also play an important role in various other robust problems, see, e.g., \cite{birge1995bounds,popescu2007robust}. 
Boshuizen and Hill \cite{boshuizen1992moment} proved a special case of Theorem \ref{thm:opt_robust_thresholds}, for the same ambiguity set as Scarf, so when the partial information consists of the mean, variance and range, and when one restricts to 
deterministic stopping rules. To explain this result in Boshuizen and Hill \cite{boshuizen1992moment} in more detail, consider Theorem~\ref{thm:opt_robust_thresholds} in the restricted setting where each offer has the same distributional information, as if the offers were generated by a sample from a large population of statistically indistinguishable buyers. For the canonical choice $\mathcal{P}(\mu,\sigma^2)$, the set of all distributions with mean $\mu$ and variance $\sigma^2$, {it can be shown \cite{boshuizen1992moment}  that nature chooses as worst-case distribution a two-point distribution with one value just below $\mu$ and one extremely large value}. This worst-case scenario forces the seller to set all thresholds equal to $\mu$, which generates an expected payoff $\mu$. Hence, while the ambiguity set $\mathcal{P}(\mu,\sigma^2)$ 
gave a nontrivial robust ordering rule in the newsvendor problem \cite{scarf1958min}, the same ambiguity set gives degenerate results in the stopping problem. Boshuizen and Hill \cite{boshuizen1992moment} countered this degeneracy by considering the ambiguity set $\mathcal{P}(\mu,\sigma^2,L)$, the set of all distributions with mean $\mu$, variance $\sigma^2$ and upper bound $L$. For this setting, Boshuizen and Hill \cite{boshuizen1992moment} found and solved the recursion in the more general Theorem \ref{thm:opt_robust_thresholds}, and found that the worst-case distributions are in fact three-point distributions (more on this later). In this paper, we also start from the observation that $\mathcal{P}(\mu,\sigma^2)$ degenerates, and we alter the information in two ways: 1) We consider two-point distributions and additionally impose an upper bound $L$ on the support, resulting in the ambiguity set $\mathcal{P}_2(\mu,\sigma^2,L)$  and 2) we replace the variance by the mean absolute deviation, another well-known measure of dispersion. 
As will become clear, both information settings will lead to two-point worst-case distributions with nondegenerate performance and hence non-trivial robust stopping policies.}

We first describe our results for  the ambiguity set $\mathcal{P}_2(\mu,\sigma^2,L)$. In this case, nature still creates a good and bad scenario, just as for $\mathcal{P}(\mu,\sigma^2)$,  but the scenarios will be less extreme due to the support upper bound $L$. We solve the maximin problem for $\mathcal{P}_2(\mu,\sigma^2,L)$  explicitly. We show that the robust max-min analysis does not become overly conservative, and instead leads to insightful optimal stopping strategies with non-trivial thresholds and expected pay-off larger than $\mu$. 
It turns out that in the case of two-point distributions there are two cases to consider. The first case, in which $L \leq 2\mu$, actually yields in every step of the backwards recursion the same worst-case two-point distribution $f^*$, with a good value $\mu + \sigma^2/\mu$ and a bad  value 0. This good-bad scenario results in the threshold-based strategy as summarized in Theorem \ref{thm:mvs_twopoint_L<2mu}.

\begin{theorem}[Ambiguity set $\mathcal{P}_2(\mu,\sigma^2,L)$ with $L \leq 2\mu$]
\label{thm:mvs_twopoint_L<2mu}
Let $n \in \N$, and let $\mu, L > 0$ and $\sigma^2 \geq 0$. Assume that $L \leq 2\mu$. 
%(implying that $\sigma \leq \mu$). 
Let $\mathcal{P}_2(\mu,\sigma^2,L)$ be the common ambiguity set consisting of all two-point distributions with mean $\mu$, variance $\sigma^2$ and support contained in $[0,L]$. For $i = 1,\dots,n-1$, the optimal robust threshold in \eqref{eq:opt_robust_tresholds} equals
\begin{align}
%\displaystyle  T(i) &= \mu \sum_{k = 0}^{n-1-i} \left( \frac{\sigma^2}{\sigma^2 + \mu^2}\right)^k \\
T(i) &= \mu + \frac{\sigma^2}{\mu}\left[1 - \left(\frac{\sigma^2}{\mu^2 + \sigma^2}\right)^{n-1-i} \right]
\label{eq:mu_sigma_L2mu}
\end{align}
and $T(n) = 0$. As $n \rightarrow \infty$, the seller's expected robust payoff $r^* = T(0)$  approaches $\mu + \sigma^2/\mu$.
\end{theorem}

 When we consider the case $L > 2\mu$, things become more involved. It is no longer the case that every step of the backwards recursion yields the same worst-case distribution for \eqref{eq:opt_robust_tresholds}.  There turn out to be two worst-case distributions, $f^*$ (same as before) and $g^*$ and a ``turning point'' $n_0$, such that $g^*$ is the worst case for $i = n-n_0,\dots,n$, and $f^*$ the worst case for $i = 1,\dots,n-n_0-1$. In another words, if $n$ is large enough we essentially end up in the same scenario as for the case $L \leq 2\mu$. Our analysis leads to the threshold-based strategy as summarized in Theorem \ref{thm:mvs_twopoint_L>2mu} for two-point distributions in combination with the assumption that $L \geq 2\mu$.

\begin{theorem}[Ambiguity set $\mathcal{P}_2(\mu,\sigma^2,L)$ with $L \geq 2\mu$]
\label{thm:mvs_twopoint_L>2mu}
Let $n \in \N$, and let $\mu, L > 0$ and $\sigma^2 \geq 0$. Assume that $L \geq 2\mu$. Let $\mathcal{P}_2(\mu,\sigma^2,L)$ be the common ambiguity set consisting of all two-point distributions with mean $\mu$, variance $\sigma^2$ and support  contained in $[0,L]$. 
There exists an $n_0 = n_0(\mu,\sigma^2,L) \in \{3,\dots,n\}$ (independent of $n$) such that  the optimal robust threshold in \eqref{eq:opt_robust_tresholds} is given by $T(n) = 0$,
\begin{align}
T(i) = L\left[1 - \left(1-\frac{\mu}{L}\right)\left[\frac{(L-\mu)^2}{(L-\mu)^2 + \sigma^2}\right]^{n-1-i} \right] 
\label{eq:mu_sigma_L>2mu_part1}
\end{align}
for $i = n - n_0 + 1,\dots,n-1$, and by
\begin{align}
T(i) = \mu + \frac{\sigma^2}{\mu}\left[1 - \left(\frac{\sigma^2}{\mu^2 + \sigma^2}\right)^{n-n_0-i} \right] + T(n-n_0+1)\left(\frac{\sigma^2}{\sigma^2 + \mu^2}\right)^{n-n_0 +1 - i}
\label{eq:mu_sigma_L>2mu_part2}
\end{align}
for $i = 1,\dots,n - n_0$. The optimal robust payoff of the seller approaches $\mu + \sigma^2/\mu$ as $n \rightarrow \infty$.
\end{theorem}
In Section \ref{sec:mvs_l>2mu} we explain how to compute $n_0$ in terms of $\mu, \sigma^2$ and $L$. We have provided an illustration of the behavior of the thresholds in Figure \ref{fig:thresholds}. Roughly speaking, our analysis in Section \ref{sec:mvs_l>2mu} shows that in every step the distributions $f^*$ and $g^*$ are the only candidates for yielding the worst-case threshold in \eqref{eq:opt_robust_tresholds} (the red squares and blue circles, respectively). Nature then chooses in every step the distribution yielding the minimum of the two (indicated by the green line). This results in the policy as described in Theorem \ref{thm:mvs_twopoint_L>2mu}. We have also indicated the turning point $n_0$ in Figure \ref{fig:thresholds}, as well as the asymptotic payoff $\mu + \sigma^2/\mu$.\\

\begin{figure}[ht!]
\centering
\includegraphics[scale=0.5]{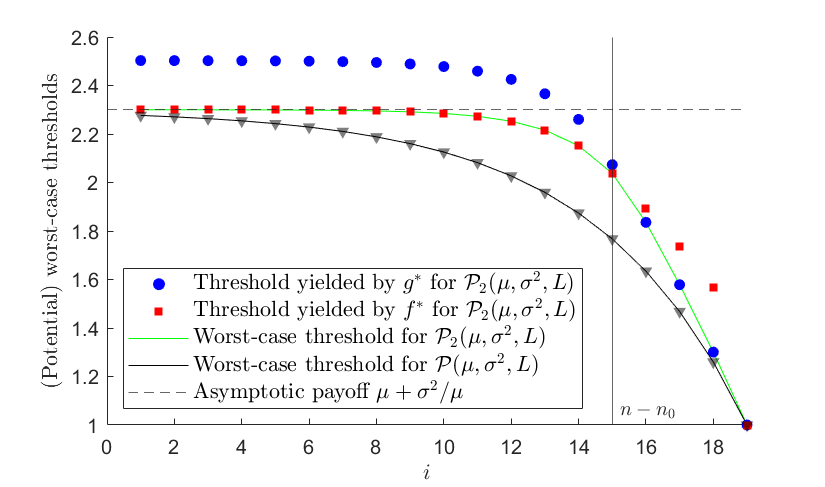}
\caption{Sketch of potential worst-case thresholds yielded by $f^*$ and $g^*$ for $\mu = 1$, $\sigma^2 = 1.3\mu$ and $L = 5\mu$. The turning point $n - n_0 = 15$ is indicated with a vertical line, and the asymptotic payoff $\mu + \sigma^2/\mu$ with a horizontal dashed line. Note that both $f^*$ and $g^*$ yield the same threshold for $i = n-1$ (namely $\mu = 1$). The robust threshold $T(i)$ as described in Theorem \ref{thm:mvs_twopoint_L>2mu} is given by the minimum in every step $i$ (indicated by the green line). That is, in steps $i > 15$ the minimum is attained by the function $g^*$, but for $i \leq 15$, it is attained by $f^*$.}
\label{fig:thresholds}
\end{figure}

In case the assumption of two-point distributions is relaxed to considering arbitrary distributions,  Boshuizen and Hill \cite[Corollary 3.2]{boshuizen1992moment} provide an analytical solution to the recursion in \eqref{eq:opt_robust_tresholds} for the ambiguity set $\mathcal{P}(\mu,\sigma^2,L)$ resulting in the thresholds 
\begin{align}
T(i) & = \mu + \frac{\sigma^2}{\mu}\left[1 - \left(1 - \frac{\mu}{L}\right)^{n-1-i} \right]
\label{eq:mvs_general}
\end{align}

{Theorems \ref{thm:mvs_twopoint_L<2mu} and \ref{thm:mvs_twopoint_L>2mu} nicely complement the result in  \cite[Corollary 3.2]{boshuizen1992moment} as for the more general ambiguity set $\mathcal{P}(\mu,\sigma^2,L)$, the worst-case distributions have a support on three points \cite{boshuizen1992moment} (more on this later).}

 {Equipped with a full understanding of  ambiguity sets based on knowning the mean, variance and support upper bound (Theorems \ref{thm:mvs_twopoint_L<2mu}, \ref{thm:mvs_twopoint_L>2mu} and \cite{boshuizen1992moment}), we further develop the connection between ambiguity and optimal stopping, by considering another related ambiguity set in which the variance in $\mathcal{P}(\mu,\sigma^2,L)$ is replaced by the mean absolute deviation (MAD). This allows use to also consider ambiguity sets for which the variance is infinite. 
The resulting ambiguity set of all distributions with mean $\mu$, mean absolute deviation $d$ and support upper bound $L$ is denoted by $\mathcal{P}(\mu,d,L)$. }

{As opposed to the ambiguity set $\mathcal{P}(\mu,\sigma^2)$, we can in fact already obtain a non-trivial solution for the set $\mathcal{P}(\mu,d)$ consisting of all distributions with mean $\mu$ and mean absolute deviation  $d$, but \emph{without} the assumption of a support upper bound $L$.  This adds significantly to the relevance of studying the MAD instead of the variance, as in this case one only has to assume knowledge of two distributional properties, rather than three.

\begin{theorem}[Ambiguity set $\mathcal{P}(\mu,d,L)$]
Let $n \in \N$,  $\mu, d, L \geq 0$. Let $\mathcal{P}(\mu,d,L)$ be the set of all distributions with mean $\mu$, mean absolute deviation $d$ and whose support is a subset of the interval $[0,L]$. Let $\mathcal{P}(I_i) = \mathcal{P}(\mu,d,L)$ for all $i = 1,\dots,n$. Then the optimal robust threshold in \eqref{eq:opt_robust_tresholds} equals
\begin{align}
T(i) = \frac{2\mu^2}{2\mu - d} - \left[\frac{2\mu^2}{2\mu - d} - \mu \right]\left( \frac{d}{2\mu}\right)^{n-1-i}
\label{eq:mms_general}
\end{align}
for $i = 1,\dots,n-1$, and $T(n) = 0$. Furthermore, as $n \rightarrow \infty$, the seller's expected robust payoff $r^* = T(0)$  approaches $2\mu^2/(2\mu-d)$.
\label{thm:mms}
\end{theorem}

\noindent Perhaps surprisingly, the thresholds in Theorem \ref{thm:mms} are also optimal for the ambiguity set $\mathcal{P}(\mu,d,L)$. That is, if one introduces an additional support constraint, then the same analysis remains valid. Furthermore, it turns out that in this case (as in Theorem \ref{thm:mvs_twopoint_L<2mu}) it is optimal for nature in every step to choose a fixed two-point distribution supported on $\{0,\mu + d\mu/(2\mu-d)\}$, which again gives rise to a good/bad scenario as the optimal strategy for nature. It is interesting to note though, that this worst-case scenario is no longer unique, leaving nature more choices.}

{In order to solve the maximin problem for the ambiguity set $\mathcal{P}(\mu,d,L)$, we use the theory of moment-bound problems. 
There is a rich literature on solving such moment problems in terms of semi-infinite linear programs (LPs) and using duality theory, see, e.g., \cite{popescu2005semidefinite} and references therein. %The novelty of our work here lies in using mathematically tractable moment problems as input for the threshold recursion in \eqref{eq:opt_robust_tresholds} to obtain explicit stopping rules.
As a proof of concept, we will explain how to use such techniques in order to prove \cite[Corollary 3.2]{boshuizen1992moment}. We show, based on a result of Cox \cite{cox1991bounds}, appearing in \cite{boshuizen1992moment} as well, that for $\mathcal{P}(\mu, \sigma^2,L)$ the worst-case distribution is in fact a three-point distribution on $\{0,T(i+1),L\}$, which in turn gives the explicit non-trivial stopping strategy in \cite[Corollary 3.2]{boshuizen1992moment}.} The thresholds in this case gradually decrease from (asymptotically) $\mu + \frac{\sigma^2}{\mu}$ to $\mu$, similar as in the analysis for $\mathcal{P}_2(\mu,\sigma^2,L)$. We have added these thresholds to Figure \ref{fig:thresholds} for comparison with the thresholds of the ambiguity set $\mathcal{P}_2(\mu,\sigma^2,L)$. We emphasize that the probability mass assigned to the points $\{0,T(i+1),L\}$ is different for every $i$. In fact, it turns out that when we move back in time, the three-point distribution gradually starts to mimic the two-point distribution $f^*$ on $\{0,\mu + \sigma^2/\mu\}$. To be more precise, the point $T(i)$ gets closer and closer to $\mu + \sigma^2/\mu$ and the probability mass put on $L$ approaches $0$. This is illustrated in Figure \ref{fig:prob_mass} in Section \ref{sec:mvs_general}. This means that for large $n$, nature will, informally speaking, eventually resort to a good/bad scenario in order to establish the worst possible situation for the seller. For both ambiguity sets $\mathcal{P}_2(\mu,\sigma^2,L)$ and $\mathcal{P}(\mu,\sigma^2,L)$, we also show (Theorems \ref{thm:asymptotic_twopoint} and \ref{thm:asymptotic}) that the seller can (almost) obtain the asymptotic payoff of $\mu + \sigma^2/\mu$ by  using the simple static threshold $T \approx \mu + \sigma^2/\mu$ in every step, if $n$ is large. In other words, in case there are many offers and the seller is not interested in choosing the optimal threshold in every individual step of the game, then this static threshold policy is a good alternative to obtain approximately the same asymptotic payoff as for the optimal policy.
 Our proof approach for Theorem \ref{thm:mms} consists of using a mean-MAD-support analogue of the result of Cox \cite{cox1991bounds}, after which we apply our recursive reasoning. More details are given in Section \ref{sec:prelim_momentbound}.

A recurrent theme in our results is the prominent role that two-point distributions play for nature. In the mean-variance-support settings, the good/bad scenario supported on $\{0,\mu + \sigma^2/\mu\}$ becomes, sooner or later, the (approximately) optimal adversarial choice for nature. The most interesting observation here seems to be the abrupt phase transition towards this distribution for the ambiguity set $\mathcal{P}_2(\mu,\sigma^2,L)$ with $L \geq 2\mu$. In the mean-MAD setting, the two-point distribution supported on $\{0,\mu + d\mu/(2\mu-d)\}$ is always an optimal adversarial choice. In conclusion, the essence of nature's adversarial response to the seller's stopping rule, under mean-dispersion information, seems to lie in creating an all-or-nothing scenario for the seller.

\subsection{Related work and further discussion}\label{sec:relwork}
The areas of stopping theory and dynamic programming are vast. For an overview of classical stopping theory, we refer the reader to Ferguson \cite{ferguson2006optimal}. 
 Various robust approaches towards stopping problems have been introduced as well. 
Closest to our work is the robust stopping framework of Riedel \cite{riedel2009optimal}, that also generalizes the classical Bayesian setting. He provides a similar result as Theorem \ref{thm:opt_robust_thresholds}, but much more general, and is based on martingale theory. However, this framework relies on some assumptions that seem not to be satisfied by our setting. We can circumvent those as we consider a simple, concrete stopping problem with independent distributions.  This allows us to give a more elementary proof of Theorem \ref{thm:opt_robust_thresholds}, avoiding advanced martingale theory, which might be of independent interest. In the context of Markov decision processes, Iyengar \cite{iyengar2005robust}, and independently Ghaoui and Nilim \cite{el2005robust}, provide a robust dynamic programming approach, leading to a robust Bellman equation, which is of a similar spirit as Theorem \ref{thm:opt_robust_thresholds}. In the Markov decision process literature, the robust perspective is introduced with respect to uncertainty arising in the transition matrix that determines with which probabilities the process moves to a new state. 

{The paper closest to our paper is 
 Boshuizen and Hill \cite{boshuizen1992moment}, as discussed earlier, on the optimal stopping problem with mean-variance information.  
  Furthermore, Stewart \cite{stewart1978optimal}, and later Samuel \cite{samuels1981minimax}, consider robust stopping rules for unknown uniform random variables, and Petruccelli \cite{petruccelli1985maximin} for normally distributed random variables. Goldenshluger and Zeevi \cite{goldenshluger2022optimal}  consider the optimal stopping problem with (fully) unknown iid distributions. }
A closely related problem to maximizing the seller's expected payoff is the so-called prophet inequality problem. Here the payoff of the seller is compared to that of a prophet, who gets to see all the offers up front and simply can select the one with the highest value. Garling, Krengel and Sucheston \cite{krengel1977} showed that the backwards induction procedure from the introduction yields an expected payoff for the seller that is at least half that of the prophet (such a result is called a prophet inequality). An elegant result by Samuel-Cahn \cite{samuel1984comparison} shows that the same guarantee can be achieved with a simple static threshold (see also the work of Kleinberg and Weinberg for an alternative static threshold approach \cite{kleinberg2012matroid}). Our asymptotic results in Theorems \ref{thm:asymptotic_twopoint} and \ref{thm:asymptotic} illustrate (albeit only in the asymptotic case) a similar contrast between the optimal strategy for the seller and simple static threshold-based alternatives. 

The prophet inequality literature has seen an immense surge in the last fifteen years; we refer to the nice survey of Lucier \cite{lucier2017survey} for an overview. In particular, various robust/limited information alternatives for the classical setting have been proposed. One prominent direction here is that initiated by Azar, Kleinberg and Weinberg~\cite{azar2014limited} where the seller only has access to samples from the underlying offer distributions, instead of the distribution itself. In this context, Rubinstein, Wang and Weinberg \cite{rubinstein2020optimal} showed that the seller can still guarantee in expectation half of the prophet's payoff in case the seller has access to one sample of every distribution. In the so-called i.i.d.~setting, in which the offer distributions are distributed identically and independently, it is even possible to obtain non-trivial prophet inequalities when nothing is known about the underlying (common) distribution of the offers, see, e.g., the work of Correa et al.~\cite{correa2021prophet_journal} and references therein.

{In the sample-access setting, the performance of the seller is measured, in addition, with respect to the expected performance under the random samples. This yields an average-case analysis with respect to the limited information. Our framework, however, uses ambiguity sets and does not require additional randomness over the given limited information. Our framework thus yields a robust worst-case analysis with respect to the limited information. Moreover, to the best of our understanding, there is no way to model sample knowledge as an ambiguity set in our framework without having to introduce additional randomness in the performance guarantee of interest (in order to make the problem well-defined and nontrivial). So whereas sample access naturally ties in with average-case analysis, our robust approach presents a worst-case analysis.}

A robust maximin approach, closer to our work, is that of D\"utting and Kesselheim \cite{dutting2019inaccurate}. They study how robust various techniques for proving prophet inequalities are, with respect to inaccurate (prior) offer distributions. In particular, they also study how errors propagate in a prophet inequality approximation guarantee when using the backwards induction procedure for computing the optimal strategy of the seller.  Further robust prophet inequality-like approaches are, e.g., considered by Schlag and Zapoechelnyuk \cite{schlag2021robust,schlag2016value}. An interesting direction for future work could be to investigate to what extent a distributionally robust prophet inequality is possible under mean-dispersion-support information. For example, as the seller can always achieve an (asymptotic) payoff of $\mu + \sigma^2/\mu$, and the prophet can clearly always get a payoff of at most $L$, the seller can obviously approximate the payoff of the prophet to within a factor $L/(\mu + \sigma^2/\mu)$. How much can this be improved?

Apart from the welfare maximization problem we study in this work, an interesting direction for future work we also plan to investigate is the revenue maximization version of the problem, in which the seller posts for every potential buyer a (randomized) price, and selects the first offer exceeding the posted price. In fact, in recent years there has been a large growth on literature concerning distributionally robust auction theory and mechanism design. This dates back to, e.g., the works of Bergemann and Schlag \cite{bergemann2008pricing,bergemann2011robust} who study a robust monopoly pricing problem. Another prominent line of work is that initiated by Azar and Micali \cite{azarmicali2013} who study parametric digital auctions under mean-variance ambiguity. It should be noted that they study the ratio between the revenue of a seller who only has access to limited distributional information, and a seller who has access to the full distributions. See also the works of Daskalakis and Pierrakos \cite{daskalakis2011simple} and Azar et al.~\cite{azar2013optimal}, and the more recent work of Giannakopoulos et al.~\cite{giannakopoulos2020robust}. The absolute revenue maximization problem is studied by Carrasco et al~\cite{carrasco2018optimal}. For further directions in distributionally robust mechanism design (some of which considering mean-dispersion information), see, e.g., \cite{bergemann2011robust,allouah2020prior,koccyiugit2020distributionally,che2019robust,he2021correlation,suzdaltsev2020distributionally,wolitzky2016mechanism,roos2019chebyshev,allouah2022pricing,anunrojwong2022robustness} and references therein.  

Our work also opens up a plethora of opportunities to study other (absolute) online selection problems with mean-dispersion information where multiple offers may be accepted by the seller subject to a combinatorial feasibility constraint. In the context of prophet inequalities, such models have been studied extensively. Here one may think of models where, e.g., the seller can accept up to $k$ offers, see, e.g., the work of Hajiaghayi et al.~\cite{hajiaghayi2007automated}. Other types of combinatorial restrictions one may consider are matroid constraints, as in the work of Kleinberg and Weinberg  \cite{kleinberg2012matroid}, or matchings constraints, see, e.g., the recent work of Ezra et al.~\cite{ezra2020online} (and references therein).
For a general combinatorial framework for prophet inequalities, we refer the reader to the work of D\"utting et al.~\cite{dutting2020prophet}.

\subsection{Technical preliminaries}	\label{sec:prelim}
	A nonnegative probability distribution is described by a probability measure $\mathbb{P} : \mathcal{B} \rightarrow [0,1]$, with $\int_x 1\, \mathrm{d}\Prob(x) $, where $\mathcal{B}$ is the Borel $\sigma$-algebra on $\R_{\geq 0}$.  We will interchangeably use the terms probability distribution and probability measure.
For a random variable $X$ distributed according to $\Prob$, denoted by $X \sim \Prob$,  and function $g: \R \rightarrow \R$, we write
$$
\E[g(X)] = \E_\Prob[g(X)] = \int_{\R_{\geq 0}} g(x) \, \mathrm{d}\Prob(x).
$$
We always assume that the integral on the right hand side is well-defined. 
We denote the \emph{mean} of $X$ by $\mu = \E[X]$, and its \emph{variance} by $\sigma^2 = \E[(X - \E[X])^2]$. The \emph{mean absolute deviation (MAD)} of $X$ is given by $d = \E\big[|X - \E[X]|\big]$, and can be seen as an $L_1$-norm dispersion measure. We say that the support of $X$ is included in the interval $[a,b] \subseteq \R_{\geq 0}$ if $\int_a^b 1\, \mathrm{d}\Prob(x) = 1$. 
In this work, for given $\mu, \sigma^2, L \geq 0$, we write $\mathcal{P}(\mu)$ for the class of all nonnegative distributions with mean $\mu$; $\mathcal{P}(\mu,\sigma^2)$ for the class of all nonnegative distributions with mean $\mu$ and variance $\sigma^2$; and $\mathcal{P}(\mu,\sigma^2,L)$ for the set of all distributions with mean $\mu$, variance $\sigma^2$ and support a subset of the interval $[0,L]$. We write $\mathcal{P}_2(\cdot) \subset \mathcal{P}(\cdot)$ for the set of all discrete distributions supported on two points (called two-point distributions) satisfying the given information $(\cdot)$. For example, $\mathcal{P}_2(\mu)$ is the set of all nonnegative two-point distributions with mean $\mu$. Analogous definitions holds for the case where the variance $\sigma^2$ is replaced by the mean absolute deviation $d \geq 0$.

We implicitly only consider combinations of the above parameters for which the resulting ambiguity sets are non-empty. In particular, for given $\mu, \sigma^2, L \geq 0$ it holds that $\mathcal{P}(\mu,\sigma^2,L)$ is non-empty if and only if
\begin{align}
\sigma^2 \leq \mu(L - \mu)
\label{eq:non-empty_sigma}
\end{align}
and $\mathcal{P}(\mu,d,L)$ is non-empty if and only if
\begin{align}
d \leq \frac{2\mu(L-\mu)}{L}
\label{eq:non-empty_mad}
\end{align}
We give a proof of the non-emptiness conditions \eqref{eq:non-empty_sigma} and \eqref{eq:non-empty_mad} in Appendix \ref{app:non-empty}. The same conditions hold true for the non-emptiness of $\mathcal{P}_2(\mu,\sigma^2,L)$ and $\mathcal{P}_2(\mu,d,L)$, respectively.

\subsection{Outline}
We continue in Section \ref{sec:robust_optimal_thresholds} with a proof of Theorem \ref{thm:opt_robust_thresholds} showing the optimality of the robust thresholds in \eqref{eq:opt_robust_tresholds}. We then provide a complete analysis of the two-point, mean-variance-support setting in Section \ref{sec:mvs_twopoint}, where we in particular proof Theorems \ref{thm:mvs_twopoint_L<2mu} and \ref{thm:mvs_twopoint_L>2mu}. In Section \ref{sec:mvs} we then explain primal-dual moment bound problems, and provide our results for the general mean-variance-support setting \cite[Corollary 3.2]{boshuizen1992moment}, as well as the mean-MAD(-support) setting (Theorem \ref{thm:mms}).
		
	\section{Optimality of robust thresholds}
	\label{sec:robust_optimal_thresholds}
	
We first formally define (randomized) stopping rules in Section \ref{sec:stopping_rules}, and afterwards give the proof of Theorem \ref{thm:opt_robust_thresholds} in Section \ref{sec:proof_main}.

\subsection{Stopping rules}\label{sec:stopping_rules}
A (randomized) stopping rule $\tau$ is a collection of functions $(\tau_i)_{i = 1,\dots,n}$, where $\tau_i$ is a (Bernoulli) random variable that takes as inputs $v_1,\dots,v_i$ and the limited information $I = (I_1,\dots,I_n)$. Based on this input, the function $\tau_i$ selects value $v_i$ with probability $r_i(v_1,\dots,v_i,I)$. To be precise, we have
\begin{align}
\tau_i(v_1,\dots,v_i,I) = \left\{
\begin{array}{ll}
1 & \text{ with prob. } r_i(v_1,\dots,v_i,I)\\
0 & \text{ with prob. } 1 - r_i(v_1,\dots,v_i,I)
\end{array}
\right.
\label{eq:tau_i}
\end{align}
for $i = 1,\dots,n$. This means that the probability we select offer $i$ is  given by
\begin{align}
\Prob[\tau = i \ | \ v_1,\dots,v_i] = r_i(v_1,\dots,v_i,I)\prod_{j=1}^{i-1}(1 - r_j(v_1,\dots,v_j,I)).
\label{eq:tau_i_prob}
\end{align}
We use $i^*$ to denote the random variable representing the stopping moment, i.e., the first $i$ for which $\tau_i(v_1,\dots,v_i,I) = 1$ in a fixed tuple $(v_1,\dots,v_n)$, which is distributed according to \eqref{eq:tau_i_prob}. We then write 
$$
X_{\tau} = \E[v_{i^*}] =  \sum_{i=1}^n \Prob[\tau = i \ | \ v_1,\dots,v_i]\cdot v_i
$$
for the expected value of the offer that the stopping rule $\tau$ selects, on the fixed tuple $(v_1,\dots,v_n)$. 
The expected payoff of the seller, where $v_i$ is a realization of a random variable $X_i \sim \Prob_i$ for $i = 1,\dots,n$, is then given by
\begin{align}
\E_{\Prob_1,\dots,\Prob_n}[X_{\tau}].
\label{eq:expected_payoff}
\end{align}

\subsection{Proof of Theorem \ref{thm:opt_robust_thresholds}}
\label{sec:proof_main}
We will use (backwards) induction. It is not hard to see that the claim is true for $n = 1$, because if there is only one offer then it is clearly optimal to simply accept it always. Now suppose the claim is true for all instances with $n-1$ offers for some $n-1 \in \N$. We prove the thresholds in \eqref{eq:opt_robust_tresholds} are also optimal for instances with $n$ offers. 

Let $\tau = (\tau_1,\dots,\tau_n)$ be a stopping rule that has higher expected robust payoff than the optimal strategy $\tau^* = (\tau_1^*,\dots,\tau_n^*)$ given by the threshold-based strategy in \eqref{eq:opt_robust_tresholds}. That is, assume that
\begin{align}
\E_{\Prob_1,\dots,\Prob_n}[X_{\tau}] > \E_{\Prob_1,\dots,\Prob_n}[X_{\tau^*}].
\label{eq:contradiction}
\end{align}
We will first argue (Lemma \ref{lem:induction_part1}), using the recursion hypothesis, that it must be the case that $\tau$ and $\tau^*$ coincide for $i = 2,\dots,n$, i.e., that $\tau$ is of the form $(\tau_1,\tau_{-1}^*) = (\tau_1,\tau^*_2\dots,\tau^*_{n})$. Secondly, we will argue (Lemma \ref{lem:induction_part2}) that then also $\tau_1 = \tau^*_1$ must hold, which gives a contradiction. 

We first introduce some additional notation. For a stopping rule $\sigma$ with fixed value $v_1$ and $\sigma_1(v_1,I) = 0$, we write $\sigma_{-1} = (\sigma_2,\dots,\sigma_n)$ for the stopping rule induced by fixing $v_1$ and $\sigma_1(v_1,I) = 0$ in $\sigma_2,\dots,\sigma_n$ (we slightly abuse notation and call the resulting stopping rules $\sigma_2,\dots,\sigma_n$ again; it will always be clear from context whether $v_1$ and $\sigma_1(v_1,I)$ are fixed or not). We remark that in, for the stopping rule $\tau^*$, fixing $v_1$ does not change $\tau_2^*,\dots,\tau_n^*$ as these are independent of $v_1$, which can be seen right away from \eqref{eq:opt_robust_tresholds}.  Note that $\sigma_{-1}$ is a well-defined stopping rule on an instance with $n - 1$ offers. Furthermore, because of the independence of the random variables, it holds that for $i = 2,\dots,n$, $v_i$ is still distributed according to $X_i \sim \Prob_i$.

\begin{lemma}
Let $\tau$ be any stopping rule, and let $\tau^*$ be the stopping rule described by the thresholds in \eqref{eq:opt_robust_tresholds}. Let $(\tau_1,\tau_{-1}^*) = (\tau_1,\tau^*_2\dots,\tau^*_{n})$. Then
\begin{align}
\min_{\Prob_1,\dots,\Prob_n} \E[X_{\tau}] \leq \min_{\Prob_1,\dots,\Prob_n} \E[X_{(\tau_1,\tau_{-1}^*)}].
\label{eq:lemma1}
\end{align}
\label{lem:induction_part1}
\end{lemma}
\begin{proof}
For any fixed realization $v_1$, it holds that
\begin{align}
\E[X_\tau \ | \ v_1] = r_1(v_1,I)\cdot v_1 + (1-r_1(v_1,I))\cdot \E[X_{\tau} \ | \ v_1, \tau_1(v_1,I) = 0].
\label{eq:fixed_v1}
\end{align}
Note that $\E[X_{\tau} \ | \ v_1, \tau_1(v_1,I) = 0]$ is the expected payoff of the seller on an instance with $n-1$ offers. The stopping rule on this instance is induced by the functions $(\tau_2,\dots,\tau_n)$ in which $v_1$ is fixed, as well as the outcome $\tau_1(v_1,I) = 0$.  The induction hypothesis then tells us that
\begin{align}
\min_{\Prob_2,\dots,\Prob_n} \E[X_{\tau_{-1}} \ | \ v_1] \leq \min_{\Prob_2,\dots,\Prob_n} \E[X_{\tau_{-1}^*} \ | \ v_1].
\label{eq:induction_part1}
\end{align}
It follows we may extend \eqref{eq:induction_part1}, using \eqref{eq:fixed_v1}, to 
\begin{align}
\min_{\Prob_2,\dots,\Prob_n} \E[X_{\tau} \ | \ v_1] \leq \min_{\Prob_2,\dots,\Prob_n} \E[X_{(\tau_1,\tau_{-1}^*)} \ | \ v_1] .
\label{eq:induction_part1b}
\end{align}
Here we use the calculus fact $\min_x f(x) \leq \min_x g(x) \Rightarrow \min_x \alpha \cdot u + (1-\alpha)f(x) \leq \min_x \alpha\cdot u + (1-\alpha)g(x)$ if $\alpha, u$ do not depend on $x$. 
We will next show how \eqref{eq:lemma1} can be derived from \eqref{eq:induction_part1b}. Fix any $\epsilon > 0$ and let $(\M_2,\dots,\M_n) \in \prod_{i=2}^n \mathcal{P}(I_i)$ be such that 
\begin{align}
\E_{\M_2,\dots,\M_n} [X_{\tau} \ | \ v_1] \leq \min_{\Prob_2,\dots,\Prob_n} \E_{\Prob_2,\dots,\Prob_n}[X_{(\tau_1,\tau_{-1}^*)} \ | \ v_1] + \epsilon.
\label{eq:inf_close}
\end{align}
Now fix some arbitrary $\G_1 \in \mathcal{P}(I_1)$. 
Integrating out \eqref{eq:inf_close} with respect to $v_1$ (using the measure $\G_1$) gives
\begin{align}
\E_{\G_1,\M_2,\dots,\M_n} [X_{\tau}] & \leq \int_{v_1} \min_{\Prob_2,\dots,\Prob_n} \E_{\Prob_2,\dots,\Prob_n}[X_{(\tau_1,\tau_{-1}^*)} \ | \ v_1] + \epsilon \, \mathrm{d}\G_1(x) \nonumber \\
& \leq \min_{\Prob_2,\dots,\Prob_n} \int_{v_1} \E_{\Prob_2,\dots,\Prob_n}[X_{(\tau_1,\tau_{-1}^*)} \ | \ v_1] \, \mathrm{d}\G_1(x) + \epsilon \nonumber \\ 
& = \min_{\Prob_2,\dots,\Prob_n} \E_{\G_1,\Prob_2,\dots,\Prob_n}[X_{(\tau_1,\tau_{-1}^*)}] + \epsilon,
\label{eq:deriv1}
\end{align}
where the second inequality is standard (and where we implicitly use the independence of the random variables $X_1,\dots,X_n$). That is, we use that for any function $g$ (for which the following expression is well-defined) it holds that $\int_t \min_x g(x,t) dt \leq \min_x \int_t g(x,t) dt$.
Taking the minimum on the left with respect to $\M_2,\dots,\M_n$ in \eqref{eq:deriv1} then gives
\begin{align}
\min_{\Prob_2,\dots,\Prob_n} \E_{\G_1,\Prob_2,\dots,\Prob_n}[X_{\tau}] \leq \E_{\G_1,\M_2,\dots,\M_n} [X_{\tau}] \leq \min_{\Prob_2,\dots,\Prob_n} \E_{\G_1,\Prob_2,\dots,\Prob_n}[X_{(\tau_1,\tau_{-1}^*)}] + \epsilon
\label{eq:deriv2}
\end{align}
Next, taking the minimum with respect to $\G_1$ in \eqref{eq:deriv2}, which was chosen arbitrarily (and recalling that $\epsilon$ was also chosen arbitrarily) then yields \eqref{eq:lemma1}. 
\end{proof}

Knowing Lemma \ref{lem:induction_part1} we continue by showing in Lemma \ref{lem:induction_part2} that it must then also hold that $\tau_1 = \tau_1^*$. This gives a contradiction with \eqref{eq:contradiction} and then completes the proof of Theorem \ref{thm:opt_robust_thresholds}.

\begin{lemma}
Let $\tau_1$ be any stopping rule for the first offer, and let $\tau^*$ be the stopping rule described by the thresholds in \eqref{eq:opt_robust_tresholds}. Let $(\tau_1,\tau_{-1}^*) = (\tau_1,\tau^*_2\dots,\tau^*_{n})$. Then
\begin{align}
\min_{\Prob_1,\dots,\Prob_n} \E[X_{(\tau_1,\tau_{-1}^*)}] \leq \min_{\Prob_1,\dots,\Prob_n} \E[X_{\tau^*}] .
\label{eq:lemma2}
\end{align}
\label{lem:induction_part2}
\end{lemma}
\begin{proof}
We have (explanation is given afterwards)
\begin{align}
\min_{\Prob_1,\dots,\Prob_n} \E[X_{(\tau_1,\tau_{-1}^*)}] &= \min_{\Prob_1} \left( \min_{\Prob_2,\dots,\Prob_n} \E[X_{(\tau_1,\tau_{-1}^*)}] \right) \nonumber \\
&= \min_{\Prob_1} \left( \min_{\Prob_2,\dots,\Prob_n} r_1(v_1,I) \cdot v_1 + (1 - r_1(v_1,I))\cdot \E[X_{(\tau_1,\tau_{-1}^*)} \ | \ \tau_1(v_1,I) = 0] \right) \nonumber \\
&= \min_{\Prob_1} \left( \min_{\Prob_2,\dots,\Prob_n} r_1(v_1,I) \cdot v_1 + (1 - r_1(v_1,I))\cdot \E[X_{\tau_{-1}^*}] \right) \nonumber \\
&= \min_{\Prob_1} \left(r_1(v_1,I) \cdot v_1 + (1 - r_1(v_1,I))\cdot \min_{\Prob_2,\dots,\Prob_n} \E[X_{\tau_{-1}^*}] \right) \nonumber \\
&= \min_{\Prob_1}\, r_1(v_1,I) \cdot v_1 + (1 - r_1(v_1,I))\cdot T(1).  
\label{eq:lemma2_equalities}
\end{align}
The first and second equality are true by definition. The third equality is a bit more subtle. Here we use the fact that the functions $\tau_2^*,\dots,\tau_n^*$ are independent of $v_1$ by definition of \eqref{eq:opt_robust_tresholds}, so once $v_1$ is rejected, we are simply left with an instance on $n-1$ offers on which the seller applies the stopping rule $\tau_{-1}^* = (\tau_2^*,\dots,\tau_n^*)$. In the fourth equality, we use the fact that $v_1$ is independent of the distributions $\Prob_2,\dots,\Prob_n$, and so we can bring the minimum inside. The final equality applies the definition of $T(1)$.

It is then not hard to see that,  independent of what distribution $\Prob_1$ is, the expression 
$$
r_1(v_1,I) \cdot v_1 + (1 - r_1(v_1,I))\cdot T(1)
$$
in \eqref{eq:lemma2_equalities} is maximized by the setting
$$
r_1(v_1,I) = \left\{
\begin{array}{ll}
1 & \text{ if } v_1 \geq T(1)\\
0 & \text{ if } v_1 < T(1)
\end{array}
\right..
$$
This is precisely the function $\tau_1^*$, which in turn yields \eqref{eq:lemma2}.
\end{proof}

\section{Ambiguity set $\mathcal{P}_2(\mu,\sigma^2,L)$} \label{sec:mvs_twopoint}
In this section we will give the proofs of Theorems \ref{thm:mvs_twopoint_L<2mu} and \ref{thm:mvs_twopoint_L>2mu} for $\Prob \in \mathcal{P}_2(\mu,\sigma^2,L)$, the ambiguity set central is this paper that contains all two-point distributions with a given mean $\mu$, variance $\sigma^2$ and support contained in $[0,L]$. We shall start with the proof of Theorem \ref{thm:mvs_twopoint_L<2mu} and then build on it in order to prove Theorem~\ref{thm:mvs_twopoint_L>2mu}. But before doing so, we {discuss} in Section~\ref{sec:m} the ambiguity sets $\mathcal{P}(\mu)$ and $\mathcal{P}(\mu,\sigma^2)$ that
 give degenerate results, {see \cite{boshuizen1992moment}}, in the sense that nature chooses worst-case two-point distributions that make an expected payoff of $\mu$ best possible. These degenerate results provide motivation to consider two-point distributions, for a slightly richer information structure that includes the support, hence $\mathcal{P}_2(\mu,\sigma^2,L)$ as discussed in Section \ref{actualproof} and further.

\subsection{Two (degenerate) ambiguity sets}\label{sec:m}

First suppose that we are only given the (common) mean $\mu$ of every distribution, i.e., $\mathcal{P}_i = \mathcal{P}(\mu)$ for all $i = 1,\dots,n$. In this case, \eqref{eq:opt_robust_tresholds} yields $T(i) = \mu$ for every $i = 0,\dots,n-1$. 
To see this, notice that the seller can guarantee an (expected) payoff of $\mu$ by simply always selecting the first offer. On the other hand, if nature chooses identical deterministic distributions (i.e., point mass of $1$ on $\mu$) for all values, then clearly $\mu$ is also the best possible; {see also \cite[Theorem 3.4]{boshuizen1992moment}.} It follows that the expected robust payoff of the seller is $r^* = \mu$ in this case. Notice that the reasoning above also remains valid if, in addition to the mean $\mu$, we are also given an upper bound $L$ on the support of the unknown distributions.

Next, assume that in addition we also know the common variance $\sigma^2$ of the distributions, i.e, we consider the ambiguity set $\mathcal{P}(\mu,\sigma^2)$ of all nonnegative distributions with mean $\mu$ and variance $\sigma^2$. Can we do better than an expected robust payoff of $\mu$? {Unfortunately not, see also \cite[Proposition 3.8]{boshuizen1992moment}.} In this case, nature can no longer simply choose deterministic distributions (point mass of $1$ on $\mu$), but we will argue that it can mimic such distributions. {Also note that, if we assume the variance is \textit{at most} $\sigma^2$, which is common in some works (e.g., \cite{carrasco2018optimal}), we can right away argue that we cannot do better, as we can then again consider the setting where nature chooses deterministic distributions with point mass at $\mu$. On the other hand, all our results seem to carry over to the setting in which we assume that the variance is \emph{at least} $\sigma^2$, but we leave this to the interested reader.}
It follows directly from \eqref{eq:opt_robust_tresholds} that $T(n-1) = \mu$. We continue next with showing that also $T(n-2) = \mu$, which then also implies that $T(i) = \mu$ for every $i$. We will show that this result can be established already with discrete two-point distributions. A two-point distribution $\Prob \in \mathcal{P}_2(\mu,\sigma^2)$ can be written as
\begin{equation}\label{eq:twopoint_general}
\Prob(x) = \left\{\begin{array}{ll}
p & \ \ \ \text{for } x = a\\
1 - p & \ \ \ \text{for } x = b
\end{array}\right.
\end{equation}
with $a \leq \mu \leq b$. The mean and variance constraint imply that (after some calculus)
\begin{equation}
a = \mu - \sqrt{\frac{1-p}{p}}\sigma \ \ \ \text{ and } \ \ \ b = \mu + \sqrt{\frac{p}{1-p}}\sigma.
\label{eq:ab}
\end{equation}
Because of the constraint $a \geq 0$, it follows that 
$$
\frac{1}{1 + c^2} \leq p < 1,
$$
where $c={\mu}/{\sigma}$ is the inverse of the \emph{coefficient of variation}.
%\label{eq:c}
%\end{equation}
Using \eqref{eq:ab} and the first form of $T(i)$ in \eqref{eq:opt_robust_tresholds}, the threshold $T(n-2)$ is then given by 
\begin{align}
T(n-2) &= \min_{\Prob \in \mathcal{P}(\mu,\sigma^2)} \int_0^{\infty} \max\{T(n-1),x\}\, \mathrm{d} \Prob(x) \nonumber \\
&= \min_{1/(1+c^2) \leq p < 1} p \mu + (1-p) b \nonumber \\
&= \min_{1/(1+c^2) \leq p < 1} p\mu + (1-p)\left(\mu + \sqrt{\frac{p}{1-p}}\sigma\right) \nonumber \\
&=  \min_{1/(1+c^2) \leq p < 1} \mu + \sqrt{p(1-p)}\sigma . 
\label{eq:twopoint_problem}
\end{align}
In order to minimize the last expression, it is not hard to see that we should take $p \rightarrow 1$, resulting in $T(n-2) = \mu$. {Alternatively, if we did not have the constraint of the random variables being nonnegative, we could also have taken $p \rightarrow 0$ (which comes at the cost of $a$ becoming negative).} So what actually happens when $p \rightarrow 1$? In this case nature's two-point distribution assigns almost all its probability mass to $\mu - g(\epsilon)$, for some $g$ with $g(\epsilon) \rightarrow 0$ as $\epsilon \rightarrow 0$, and its remaining mass to the number $b(\epsilon)$ with $b(\epsilon) \rightarrow \infty$ as $\epsilon \rightarrow 0$ (to account for the variance constraint). In other words, it assigns most of its mass to a number just below the mean $\mu$, and its remaining mass to a very large number.

The construction above heavily exploits the fact that we can make the number $b$ in the support of the two-point distribution $\Prob$ arbitrarily large. This motivates the introduction of an upper bound on the support of the unknown distributions, i.e., the study of the ambiguity set $\mathcal{P}(\mu,\sigma^2,L)$.

\subsection{Proof of Theorem \ref{thm:mvs_twopoint_L<2mu}}\label{actualproof}
 As before we use that  a two-point distribution $\Prob \in \mathcal{P}_2(\mu,\sigma^2,L)$ can be written as
 \eqref{eq:twopoint_general} and \eqref{eq:ab} 
but now the constraints $a \geq 0$ and $b \leq L$ imply that $p$ is bounded by
\begin{align}
\frac{1}{1+c^2} \leq p \leq \frac{1}{1 + \frac{\sigma^2}{(L-\mu)^2}} \ \  \left( \leq 1 - \frac{1}{1 + c^2} \text{ as } L \leq 2\mu \right),
\label{eq:p_inequalities}
\end{align}
where, again, $c = \mu/\sigma$. The two extreme choices for $p$ yield the distributions $f^*$ and $g^*$ that were used in Section \ref{sec:results}. We remark that the interval of $p$'s satisfying the bounds in \eqref{eq:p_inequalities} is non-empty as for any distribution supported on $[0,L]$ with mean $\mu$ and variance $\sigma^2$, it holds that $\sigma^2 \leq \mu(L - \mu)$, because of \eqref{eq:non-empty_sigma}, which is equivalent to $1/(1+c^2) \leq 1/(1 + (\sigma/(L - \mu))^2)$. Furthermore, the assumption $L \leq 2\mu$ implies that $\sigma \leq \mu$, which yields $c  \geq 1$.

Again we have $T(n) = 0$ and $T(n-1) = \mu$. 
Using the same reasoning as in \eqref{eq:twopoint_problem}, we now find the following minimization problem for $p$:
\begin{eqnarray}
T(n-2) =  & \min & \mu + \sqrt{p(1-p)}\sigma  \nonumber  \\
 & \text{s.t.} & \frac{1}{1+c^2}\leq \ p \ < \frac{1}{1 + \frac{\sigma^2}{(L-\mu)^2}}. 
 \label{eq:T(n-2)}
\end{eqnarray}

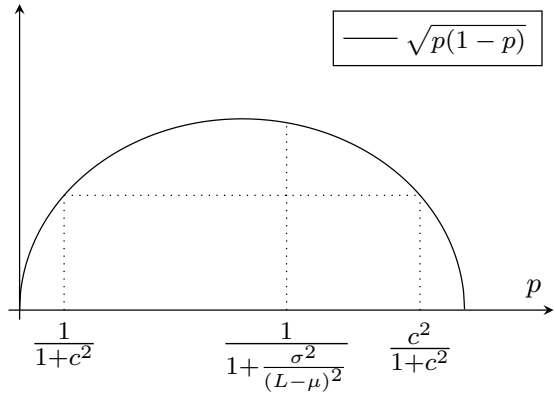
\begin{figure}[ht!]
\centering
\begin{tikzpicture}[scale=1.2]
%Removed 'axis equal'
    \begin{axis}[axis lines=middle,xlabel=\footnotesize{$p$},ylabel={},
        domain=0:1,xmin=-0.025,xmax=1.2,ymin=-0.025,ymax=0.8,xtick=\empty,ytick=\empty,
        smooth,samples=500,clip=false,width=3in,height=2in,
        legend style = {font = {\fontsize{8 pt}{12 pt}\selectfont}}]
        %\draw[dotted] (-2.5,-2) -- (-2.5,2);
        %\draw (-2.5,0) node[below]{-5};
        \draw[dotted] (0.1,0) -- (0.1,0.3);
        \draw (0.1,0) node[below]{$\frac{1}{1+c^2}$};
        \draw[dotted] (0.9,0) -- (0.9,0.3);
        \draw (0.9,0) node[below]{$\frac{c^2}{1+c^2}$};
        \draw[dotted] (0.1,0.3) -- (0.9,0.3);
        \draw[dotted] (0.6,0) -- (0.6,0.49);
        \draw (0.6,0) node[below]{$\frac{1}{1+\frac{\sigma^2}{(L-\mu)^2}}$};
        \addplot[color=black] {(x*(1-x))^(0.5)}
            node[pos=1,sloped,below left] {};
        \addlegendentry{\(\sqrt{p(1-p)}\)};
        %\addplot[color=blue] {3-x*(exp(x)+1)/(exp(x)-1)}
        %    node[pos=1,sloped,above left] {$y=f(x)$};
 \end{axis}         
\end{tikzpicture}
\caption{Sketch of where minimum of $h(p) = \sqrt{p(1-p)}$ is attained for $T(n-2)$. }
\label{fig:p-curve}
\end{figure}

By symmetry of the function $h(p) = \sqrt{p(1-p)}$ around $p = 1/2$, i.e., $h(1/(1+c^2)) = h(1 - 1/(1+c^2))$, it then follows that the minimum of the function $h(p)$ is attained at $p = 1/(1+c^2)$. Note that here we also use the fact that $1/(1+c^2) \leq 1/2$ as $c \geq 1$. An illustration of the situation is given in Figure \ref{fig:p-curve}.  We can now determine the optimal robust threshold $T(n-2)$  by taking $p = 1/(1+c^2)$ in \eqref{eq:T(n-2)}. 
We find
$$
\displaystyle T(n-2) = \left(1 + \frac{1}{1 + c^2}\right)\mu.
$$
The worst-case distribution $\Prob_{n-1}$ attaining the minimum is given by 
\begin{align}
\Prob_{n-1}(x) = \left\{\begin{array}{ll}
\displaystyle \frac{1}{1+c^2} & \ \ \ \text{for } \displaystyle  x = a = 0 \\
\displaystyle 1 - \frac{1}{1+c^2} & \ \ \ \text{for } \displaystyle  x = b = \left(1 + \frac{1}{c^2}\right)\mu 
\end{array}\right. .
\label{eq:worst_f}
\end{align}
The point $b = \left(1 + 1/c^2\right)\mu = \mu + \sigma^2/\mu$ plays a special role, as it is the smallest possible rightmost point on which a distribution $\Prob \in \mathcal{P}_2(\mu,\sigma^2,L)$ is supported. The following lemma will also play an important role (also later on).
\begin{lemma}
The sequence $(T(i))_{i=0,\dots,n}$ described by \eqref{eq:opt_robust_tresholds} is nonincreasing for any collection of ambiguity sets $\mathcal{P}(I_i)$. Furthermore, suppose that every ambiguity set contains the two-point distribution in \eqref{eq:worst_f}.
Then $\mu = T(n-1) \leq \dots \leq T(1) \leq T(0) \leq \mu + \sigma^2/\mu$.
\label{lem:bound}
\end{lemma}
\begin{proof}
The fact that the sequence $(T(i))_{i=0,\dots,n}$ is nonincreasing follows directly from the definition
$$
T(i) = \min_{\Prob_{i+1} \in \mathcal{P}(I_{i+1})} \int_0^{\infty} \max\{T(i+1),x\}\, \mathrm{d} \Prob_{i+1}(x),
$$
as $\max\{T(i+1),x\} \geq T(i+1)$.
In order to prove the second claim, it suffices to consider the case $i = 0$, because we have just shown that the sequence $(T(i))_{i=0,\dots,n}$ is nonincreasing. Suppose that nature chooses the two-point distribution \eqref{eq:worst_f} for every offer $i$. Then clearly the seller can never obtain an expected robust payoff of more than $\mu + \sigma^2/\mu$ as the value $v_i$ of every offer is at most $\mu + \sigma^2/\mu$ (i.e., this is the largest point in the support of $\Prob$).
\end{proof}

From Lemma \ref{lem:bound}, we therefore can infer that for every $\Prob$ of the form \eqref{eq:twopoint_general} and every $i = 0,\dots,n-1$, it holds that 
$$
a \leq \mu \leq T(i) \leq \mu + \sigma^2/\mu \leq b.
$$
This means the integral in \eqref{eq:opt_robust_tresholds} is equal to
\begin{align}
T(i) &= \displaystyle \min_{\Prob_{i+1} \in \mathcal{P}_2(\mu,\sigma^2,L)} \int_0^{\infty} \max\{T(i+1),x\}\, \mathrm{d} \Prob_{i+1}(x) \nonumber  \\
& = \displaystyle \min_{\frac{1}{1+c^2}\leq \ p \ < \frac{1}{1 + \frac{\sigma^2}{(L-\mu)^2}}} T(i+1)p + (\mu + \sqrt{p/(1-p)}\sigma)(1-p) \nonumber \\
& = \displaystyle \min_{\frac{1}{1+c^2}\leq \ p \ < \frac{1}{1 + \frac{\sigma^2}{(L-\mu)^2}}} \mu + (T(i+1) - \mu)p + \sqrt{p(1-p)}\sigma
\label{eq:twopoint_rec_L<2mu}
\end{align}
for $i = 0,\dots,n-1$. In order to prove that the minimum in \eqref{eq:twopoint_rec_L<2mu} remains to be attained at $p = 1/(1+c^2)$, we will use the following proposition.

\begin{proposition}
Let $K_0, K_1 \geq 0$ and $p \in (0,1)$. Assume that $0 < y \leq 1/2$ and $z$ is such that $y \leq z \leq 1 - y$. Let 
$$
u(p) = K_0p + \sqrt{p(1-p)}K_1.
$$
Then $u(y) \leq \min\{u(z),u(1 - y)\}$ and the minimum of $u$ in the range $(y,1-y)$ is attained at $y$.
\label{prop:p}
\end{proposition}
\begin{proof}
We write $u(p) = u_1(p) + u_2(p)$ with $u_1(p) = K_0p$ and $u_2(p) = \sqrt{p(1-p)}K_1$. Note that $u_2(y) = u_2(1-y)$ by symmetry of $u_2$ around $p = 1/2$. Furthermore, $u_1(y) \leq u_2(1-y)$ as $y \leq 1- y$ by assumption of $0 \leq y \leq 1/2$. It follows that $u(y) \leq u(1-y)$.

Secondly, the function $u(p)$ has at most one point in the interval $(0,1)$ where its derivative vanishes for any choice of $K_0$ and $K_1$. In combination with the fact that $u(p)$ has a positive derivative for any $0 < p < 1/2$, and $u(y) \leq u(1 - y)$, it then follows that for every point $y \leq z \leq 1 - y$ satisfies $u(y) \leq u(z)$ and that the minimum of $u$ in $(y,1-y)$ must be attained at $y$. 
\end{proof}

The claim that $p = 1/(1+c^2)$ is where the minimum of \eqref{eq:twopoint_rec_L<2mu} is attained now follows from applying Proposition \ref{prop:p} with $y = 1/(1+c^2) \leq 1/2$, $K_0 = T(i+1) - \mu \geq 0$ and $K_1 = \sigma \geq 0$. We emphasize that this holds for any $i$. Plugging $p = 1/(1+c^2)$ into the expression \eqref{eq:twopoint_rec_L<2mu} gives
\begin{align}
T(i) = \mu + (T(i+1) - \mu)\frac{1}{1+c^2} + \frac{c}{1+c^2}\sigma = \mu + \frac{1}{1+c^2}T(i+1).
\label{eq:recursion_L<2mu}
\end{align}

Here, and in subsequent sections, we will repeatedly make use of backwards recursions of this type. We now give a general description and solution of such recursions. The proof of Proposition \ref{prop:rec} is left to the reader.
	
\begin{proposition}
\label{prop:rec}
Let $\alpha, \beta, \gamma_0 \in \R$, and consider the backwards recursion
$
t(i) = \alpha + \beta \cdot t(i+1)
$
for $i = 0,\dots,m$ with $t(m) = \gamma_0$. The solution to this recursion is given by 
\begin{align}
t(i) = \left(\alpha \sum_{k=0}^{m - 1 - i} \beta^k\right) + \gamma_0 \beta^{m - i} 
%= \frac{\alpha}{1-\beta} - \left[\frac{\alpha}{1-\beta} - \gamma_0\right]\beta^{m-i}
\label{eq:rec_general_sol}
\end{align}
for $i = 0,\dots,m$. {We use the convention that $\sum_{k=0}^{-1} \beta^k = 0$.} 
\end{proposition}

Applying Proposition \ref{prop:rec} with $\alpha = \gamma_0 = \mu$ and $\beta = 1/1(1+c^2)$ then yields
$$
T(i) = \mu \sum_{k=0}^{n-1-i} \left(\frac{1}{1+c^2}\right)^k = \mu \sum_{k = 0}^{n-1-i} \left( \frac{\sigma^2}{\sigma^2 + \mu^2}\right)^k.
$$
Using the identity $\sum_{k=0}^j x^j = (1 - x^{j+1})/(1-x)$, in combination with some calculus, then yields the identity as stated in Theorem \ref{thm:mvs_twopoint_L<2mu}. This completes the proof of Theorem \ref{thm:mvs_twopoint_L<2mu}.

\subsection{Proof of Theorem \ref{thm:mvs_twopoint_L>2mu}}
\label{sec:mvs_l>2mu}
The setup for the proof of Theorem \ref{thm:mvs_twopoint_L>2mu} is similar to that of the proof of Theorem \ref{thm:mvs_twopoint_L<2mu}. We again have $T(n) = 0$ and $T(n-1) = \mu$. The difference occurs when we consider the minimization problem 
\begin{eqnarray}
T(i) =  & \displaystyle \min&  \mu + (T(i+1) - \mu)p + \sqrt{p(1-p)}\sigma.  \nonumber  \\
 & \text{s.t.} & \frac{1}{1+c^2}\leq \ p \ < \frac{1}{1 + \frac{\sigma^2}{(L-\mu)^2}}
 \label{eq:T(n-2)_L>2mu}
\end{eqnarray}
for $i = 0,\dots,n-2$. 

Because of the assumption $L > 2\mu$, it already holds for $i = n-2$ that the minimum is not attained at the leftmost point $p = 1/(1+c^2)$, but at the rightmost point $p = 1/(1 + \sigma^2/(L-\mu)^2)$ instead. {For certain parameter combinations the minimum is attained at both points simultaneously.} To see this, one has to distinguish two cases: either $c \geq 1$, in which case the condition $L > 2\mu$ yields the claim, 
%\footnote{As then $1/(1 + \sigma^2/(L-\mu)^2) \geq 1 - 1/(1+c^2)$ and so $h(1/(1 + \sigma^2/(L-\mu)^2)) \leq h(1 - 1/(1+c^2)) = h(1/(1+c^2))$.}
or $c < 1$, in which case $1/(1+c^2) \geq 1/2$ and so $h(1/(1+c^2)) \geq h(1/(1 + \sigma^2/(L-\mu)^2))$ as $h(p) = \sqrt{p(1-p)}$ is decreasing on $[1/2,1]$. These cases are illustrated in Figure \ref{fig:p-curves}.
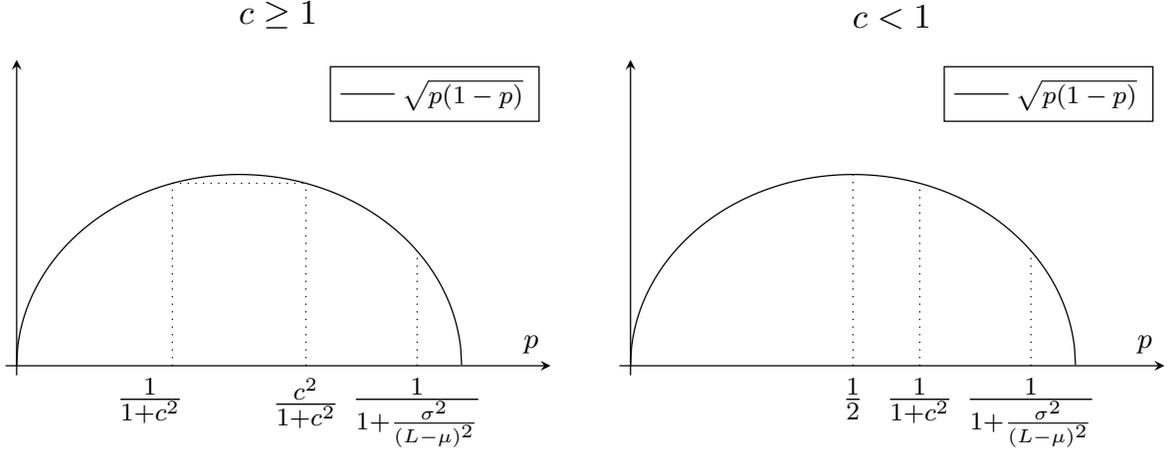
\begin{figure}[ht!]
\centering
\begin{tikzpicture}[scale=1.2]
    \begin{axis}[title={$c \geq 1$},axis lines=middle,xlabel=\footnotesize{$p$},ylabel={},
        domain=0:1,xmin=-0.025,xmax=1.2,ymin=-0.025,ymax=0.8,xtick=\empty,ytick=\empty,
        smooth,samples=500,clip=false,width=3in,height=2in,
        legend style = {font = {\fontsize{8 pt}{12 pt}\selectfont}}]
        %\draw[dotted] (-2.5,-2) -- (-2.5,2);
        %\draw (-2.5,0) node[below]{-5};
        \draw[dotted] (0.35,0) -- (0.35,0.477);
        \draw (0.3,0) node[below]{$\frac{1}{1+c^2}$};
        \draw[dotted] (0.65,0) -- (0.65,0.477);
        \draw (0.65,0) node[below]{$\frac{c^2}{1+c^2}$};
        \draw[dotted] (0.35,0.477) -- (0.65,0.477);
        \draw[dotted] (0.9,0) -- (0.9,0.3);
        \draw (0.9,0) node[below]{$\frac{1}{1+\frac{\sigma^2}{(L-\mu)^2}}$};
        \addplot[color=black] {(x*(1-x))^(0.5)}
            node[pos=1,sloped,below left] {};
        \addlegendentry{\(\sqrt{p(1-p)}\)};
        %\addplot[color=blue] {3-x*(exp(x)+1)/(exp(x)-1)}
        %    node[pos=1,sloped,above left] {$y=f(x)$};
 \end{axis}         
\end{tikzpicture} \quad \quad
\begin{tikzpicture}[scale=1.2]
    \begin{axis}[title={$c <  1$},axis lines=middle,xlabel=\footnotesize{$p$},ylabel={},
        domain=0:1,xmin=-0.025,xmax=1.2,ymin=-0.025,ymax=0.8,xtick=\empty,ytick=\empty,
        smooth,samples=500,clip=false,width=3in,height=2in,
        legend style = {font = {\fontsize{8 pt}{12 pt}\selectfont}}]
        %\draw[dotted] (-2.5,-2) -- (-2.5,2);
        %\draw (-2.5,0) node[below]{-5};
        \draw[dotted] (0.5,0) -- (0.5,0.5);
        \draw (0.5,0) node[below]{$\frac{1}{2}$};
        \draw[dotted] (0.65,0) -- (0.65,0.477);
        \draw (0.65,0) node[below]{$\frac{1}{1+c^2}$};
        %\draw[dotted] (0.35,0.477) -- (0.65,0.477);
        \draw[dotted] (0.9,0) -- (0.9,0.3);
        \draw (0.9,0) node[below]{$\frac{1}{1+\frac{\sigma^2}{(L-\mu)^2}}$};
        \addplot[color=black] {(x*(1-x))^(0.5)}
            node[pos=1,sloped,below left] {};
        \addlegendentry{\(\sqrt{p(1-p)}\)};
        %\addplot[color=blue] {3-x*(exp(x)+1)/(exp(x)-1)}
        %    node[pos=1,sloped,above left] {$y=f(x)$};
 \end{axis}         
\end{tikzpicture}
\caption{Sketch of where minimum of $h(p) = \sqrt{p(1-p)}$ is attained for $T(n-2)$ depending on whether $c \geq 1$ (left) or $c < 1$ (right). }
\label{fig:p-curves}
\end{figure}

The main question is now as follows: What happens for $T(i)$ for $i = 0,\dots,n-3$? Is the minimum in \eqref{eq:T(n-2)_L>2mu} still attained at the rightmost point $p = 1/(1 + \sigma^2/(L-\mu)^2)$ or will the minimum be attained at the leftmost point $p = 1/(1+c^2)$? The reason that we may focus on the extreme points of $p$ is again because at every step of the backward recursion we are concerned with a function of the form $u(p)$ as in Lemma \ref{prop:p}, which attains its minimum over an interval $[c,d] \subseteq [0,1]$ at the boundary point $c$ or $d$. It does not seem easy to give an explicit answer to this question in terms of the parameters $\mu$, $\sigma^2$ and $L$ right away. Nevertheless, we are able to describe the behavior of the minimum over the course of the backward recursion. We emphasize that this behavior relies on the fact that $L$ is a constant independent of $n$. 

We will show that there exists an $n_0 = n_0(\mu,\sigma^2,L) \in \{3,\dots,n\}$ such that:
\begin{itemize}
\item For $n-n_0 < i \leq n-2$, the minimum in  \eqref{eq:T(n-2)_L>2mu} will be attained at the rightmost point $p = 1/(1 + \sigma^2/(L-\mu)^2)$.
\item For $1 \leq i \leq n - n_0$, the minimum will be attained at the leftmost point $p = 1/(1+c^2)$.
\end{itemize}
Note that this means that if $n_0 = n$, then in fact the minimum is always attained at the rightmost point $p$. In other words, if $n$ is large enough, there is a critical point $n - n_0$ at which the minimum switches from the rightmost point to the leftmost point, and after this, the leftmost point will continue to be the minimum for the remaining steps of the backwards recursion. We continue with the proof of this claim.

 First of all, let $n_0$ be such that $n - n_0$ is the highest index $i$ for which the minimum in \eqref{eq:T(n-2)_L>2mu} is attained at the leftmost point $p = 1/(1+c^2)$ (or $n_0 = n$ otherwise). This means that for $n - n_0 + 1 \leq i \leq n-2$ the minimum is attained at the rightmost point $p = 1/(1 + \sigma^2/(L-\mu)^2)$. Based on \eqref{eq:T(n-2)_L>2mu}, the thresholds for these indices $i$ are then given by the solution to the recursion
\begin{eqnarray}
T(i) & = &\mu + (T(i+1) - \mu)p + \sqrt{p(1-p)}\sigma \nonumber \\
 & = &\mu + (T(i+1) - \mu)\frac{1}{1 + \frac{\sigma^2}{(L-\mu)^2}} + \sqrt{\frac{1}{1 + \frac{\sigma^2}{(L-\mu)^2}}\left(1-\frac{1}{1 + \frac{\sigma^2}{(L-\mu)^2}}\right)} \, \sigma \nonumber  \\
 & = & \frac{L\sigma^2}{(L-\mu)^2 + \sigma^2} + \frac{(L-\mu)^2}{(L-\mu)^2 + \sigma^2}\, T(i+1).
 \label{eq:recursion_L>2mu}
\end{eqnarray}
Using Proposition \ref{prop:rec} with
$$
\alpha = \frac{L\sigma^2}{(L-\mu)^2 + \sigma^2}, \ \ \beta = \frac{(L-\mu)^2}{(L-\mu)^2 + \sigma^2}, \ \ \text{ and } \ \ \gamma_0 = \mu
$$
then yields for $n - n_0 + 1 \leq i \leq n-1$ that
\begin{align}
T(i) &= \frac{L\sigma^2}{(L-\mu)^2 + \sigma^2} \left[\sum_{k=0}^{n - 2 - i} \left(\frac{(L-\mu)^2}{(L-\mu)^2 + \sigma^2} \right)^k\right] + \mu \left( \frac{(L-\mu)^2}{(L-\mu)^2 + \sigma^2}\right)^{n-1-i} \nonumber \\
& = L\left[1 - \left(1-\frac{\mu}{L}\right)\left[\frac{(L-\mu)^2}{(L-\mu)^2 + \sigma^2}\right]^{n-1-i} \right].
\label{eq:T(i)_last}
\end{align}

\noindent Secondly, recall that $i = n - n_0$ is the first time the minimum is attained at the leftmost point $p = 1/(1+c^2)$.  We claim that then in all the remaining steps $i = 1,\dots,n-n_0-1$ of the backwards recursion, the minimum will also be attained at the leftmost point $p = 1/(1+c^2)$. 
To see this, note that the recursion, where $T(i)$ should be seen as a function of $p$, for $k < n - n_0$ reads 
\begin{align}
T(i) &= \mu + (T(i+1) - \mu)p + \sqrt{p(1-p)}\sigma \nonumber \\
&= (T(i+1) - T(i+2))p + \big[\mu + (T(i+2) - \mu)p + \sqrt{p(1-p)}\sigma\big] \nonumber \\
&=: g_1^i(p) + g_2^i(p)
\end{align}
with $g_1^i(p) = (T(i+1) - T(i+2))p$ and $g_2^i(p) = \mu + (T(i+2) - \mu)p + \sqrt{p(1-p)}\sigma$. Now suppose that the minimum of $T(i)$ is attained at $p = 1/(1+c^2)$ for $i = \ell + 1,\dots,n-n_0$. We claim that the same holds for $i = \ell$. To see this, note that $g_2^i(p)$ is precisely the expression we minimized for $i = \ell + 1$, so it follows that  $g_2^i(1/(1+c^2)) \leq g_2^i(1/(1 + \sigma^2/(L-\mu)^2))$ as we assume the minimum is attained at $p = 1/(1+c^2)$ for $i = \ell + 1$. Furthermore, we know from Lemma \ref{lem:bound} that the optimal robust thresholds in \eqref{eq:opt_robust_tresholds} form a non-increasing sequence. This means that $g_1^i(p)$ is an non-decreasing function in $p$, and therefore $g_1^i(1/(1+c^2)) \leq g_1^i(1/(1 + \sigma^2/(L-\mu)^2))$. Hence the function $T(\ell)$, as a function of $p$, is minimized at $p = 1/(1+c^2)$ as well. 

A sketch of the behavior of the minimum is given in Figure \ref{fig:min-switch}. Roughly speaking, what happens is that the top of the function $T(i) = g_1^i(p) + g_2^i(p)$ moves to the right as $i$ decreases (i.e., as we go further into the backward recursion). This means that as soon as the minimum is attained on the left it will be attained on the left in the future as well.

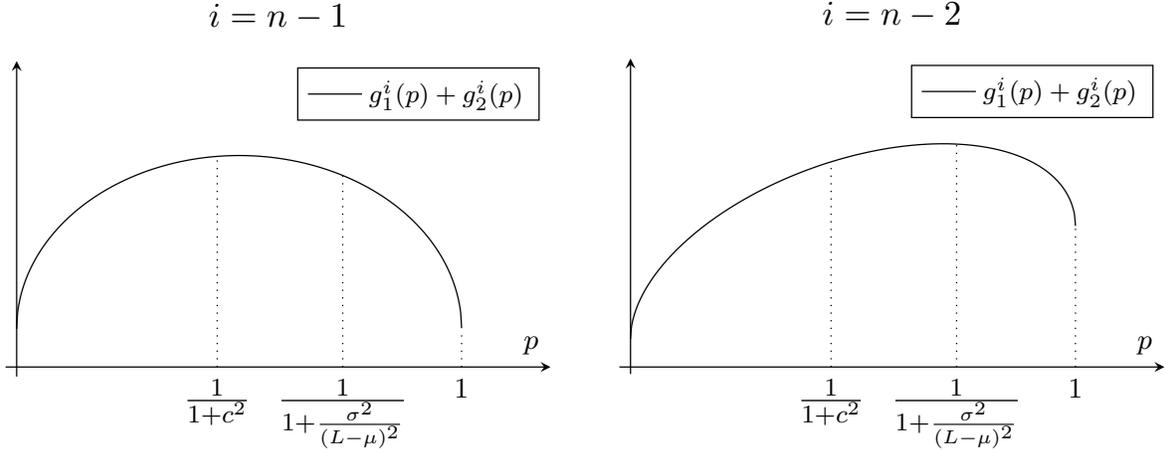
\begin{figure}[ht!]
\centering
\begin{tikzpicture}[scale=1.2]
    \begin{axis}[title={$i = n - 1$},axis lines=middle,xlabel=\footnotesize{$p$},ylabel={},
        domain=0:1,xmin=-0.025,xmax=1.2,ymin=-0.025,ymax=0.8,xtick=\empty,ytick=\empty,
        smooth,samples=500,clip=false,width=3in,height=2in,
        legend style = {font = {\fontsize{8 pt}{12 pt}\selectfont}}]
        %\draw[dotted] (-2.5,-2) -- (-2.5,2);
        %\draw (-2.5,0) node[below]{-5};
        \draw[dotted] (\pleft,0) -- (\pleft,\Tpleft-0.9);
        \draw (\pleft,0) node[below]{$\frac{1}{1+c^2}$};
        \draw[dotted] (1,0) -- (1,0.1);
        \draw (1,0) node[below]{\footnotesize{$1$}};
        %\draw[dotted] (0.35,0.477) -- (0.65,0.477);
        \draw[dotted] (\pright,0) -- (\pright,\Tpright-0.9);
        \draw (\pright,0) node[below]{$\frac{1}{1+\frac{\sigma^2}{(L-\mu)^2}}$};
        \addplot[color=black] {\m-0.9 + \s*(x*(1-x))^(0.5)}
            node[pos=1,sloped,below left] {};
        \addlegendentry{\(g_1^i(p) + g_2^i(p)\)};
        %\addplot[color=blue] {3-x*(exp(x)+1)/(exp(x)-1)}
        %    node[pos=1,sloped,above left] {$y=f(x)$};
 \end{axis}         
\end{tikzpicture} \quad \quad
\begin{tikzpicture}[scale=1.2]
    \begin{axis}[title={$i = n - 2$},axis lines=middle,xlabel=\footnotesize{$p$},ylabel={},
        domain=0:1,xmin=-0.025,xmax=1.2,ymin=-0.025,ymax=1.1,xtick=\empty,ytick=\empty,
        smooth,samples=500,clip=false,width=3in,height=2in,
        legend style = {font = {\fontsize{8 pt}{12 pt}\selectfont}}]
        %\draw[dotted] (-2.5,-2) -- (-2.5,2);
        %\draw (-2.5,0) node[below]{-5};
        \draw[dotted] (\pleft,0) -- (\pleft,\Ttwopleft-0.9);
        \draw (\pleft,0) node[below]{$\frac{1}{1+c^2}$};
        \draw[dotted] (1,0) -- (1,\Tpright-0.9);
        \draw (1,0) node[below]{\footnotesize{$1$}};
        %\draw[dotted] (0.35,0.477) -- (0.65,0.477);
        \draw[dotted] (\pright,0) -- (\pright,\Ttwopright-0.9);
        \draw (\pright,0) node[below]{$\frac{1}{1+\frac{\sigma^2}{(L-\mu)^2}}$};
        \addplot[color=black] {\m-0.9 + (\Tpright - \m)*x + \s*(x*(1-x))^(0.5)}
            node[pos=1,sloped,below left] {};
        \addlegendentry{\(g_1^i(p) + g_2^i(p)\)};
 \end{axis}         
\end{tikzpicture}
%\includegraphics[scale=0.65]{minimum_switch_V2.png}
%\caption{\jvl{Deze plaatjes in dezelfde stijl brengen als eerdere twee plaatje, en naast elkaar plotten}
\caption{Behavior of minimum of $T(i) = g_1^i(p) + g_2^i(p)$ for $i = n-2$ (left) and $i = n-3$ (right). For $i = n-1$ the minimum is attained at the rightmost feasible $p = 1/(1 + \sigma^2/(L-\mu)^2)$, but for $i = n-2$ it is attained at the leftmost feasible $p = 1/(1+c^2)$. These plots are made with $\mu = 1$, $\sigma^2 = 0.82\mu$ and $L = 2.5\mu$.}
\label{fig:min-switch}
\end{figure}

Finally, we can again solve the recursion using Proposition \ref{prop:rec} and obtain 
\begin{align}
T(i) &= \left(\sum_{k=0}^{n - n_0 - i} \left(\frac{\sigma^2}{\mu^2+\sigma^2}\right)^i \right)\mu + \left(\frac{\sigma^2}{\mu^2+\sigma^2}\right)^{n - n_0 + 1 - i}T(n - n_0 + 1) \nonumber \\
&= \mu + \frac{\sigma^2}{\mu}\left[1 - \left(\frac{\sigma^2}{\mu^2 + \sigma^2}\right)^{n-n_0-i} \right] + T(n-n_0+1)\left(\frac{\sigma^2}{\sigma^2 + \mu^2}\right)^{n-n_0 +1 - i}
\label{eq:T(i)_first}
\end{align}
for $i = 1,\dots, n-n_0$, where $T(n-n_0+1)$ is given by \eqref{eq:T(i)_last}.\\

\noindent In order to understand what happens when $n \rightarrow \infty$, we carry out a little thought experiment. Suppose that $n_0 = 1$, even if $n \rightarrow \infty$, meaning that the minimum is attained at $p = 1/(1+\sigma^2/(L-\mu)^2)$ for every step of the backwards recursion. This then means that $T(0) \rightarrow L$ as $n \rightarrow \infty$, as can be seen from \eqref{eq:T(i)_last}. However, we know that nature can always guarantee that the seller gets at most $\mu + \sigma^2/\mu$ by Lemma \ref{lem:bound}. Note that in fact $\mu + \sigma^2/\mu \leq L$, as this is equivalent to
$
\sigma^2 \leq \mu(L - \mu)
$
which is true for any distribution with mean $\mu$, variance $\sigma^2$ and support range in $[0,L]$. This means that, when the strict inequality $\sigma^2 < \mu(L-\mu)$ holds, it must be the case that at some point $n - n_0$ in the recursion the minimum will be attained at the leftmost point $p = 1/(1 + c^2)$ for the first time. This means we will always switch to the thresholds as in \eqref{eq:T(i)_first}. These have the property that $T(0) \rightarrow \mu + \sigma^2/\mu$ as $n \rightarrow \infty$, regardless of $n_0 = n_0(\mu,\sigma^2,L)$, which itself is independent of $n$.

In order to compute the value of $n_0$ exactly, one has to compare the values of $T(i)$ obtained by either applying the recursion in \eqref{eq:recursion_L<2mu} or \eqref{eq:recursion_L>2mu} in every step $i$. The minimum of these two values determines the worst-case distribution (and  the robust threshold).  We know, by the analysis above, that in the first steps the minimum will be given by \eqref{eq:recursion_L<2mu}. The first step $i$ for which the minimum is given by \eqref{eq:recursion_L>2mu}, determines the value of $n_0$.

\subsection{An asymptotic result}
\label{sec:asymptotic_twopoint}
Both in the case $L \leq 2\mu$ (Theorem \ref{thm:mvs_twopoint_L<2mu}) and in the case $L > 2\mu$ (Theorem \ref{thm:mvs_twopoint_L>2mu}) it follows that as $n \rightarrow \infty$, the expected robust payoff of the optimal stopping strategy of the gambler approaches $\mu + \sigma^2/\mu$. 
We argue that we can also obtain this payoff approximately (as $n$ grows large) by means of simply setting a threshold of 
\begin{equation}
T = \mu + \frac{\sigma^2}{\mu}
\end{equation}
at every step of the game. That is, if one is only interested in obtaining the optimal payoff when $n$ is large, and not so much about identifying the optimal strategy at every individual step of the game, then this can be done in a fairly simple manner. A slightly weaker version of this result for the more general set $\mathcal{P}(\mu,\sigma^2,L)$ is given in Section \ref{sec:asymptotic}.
\begin{theorem}
Consider the stopping rule $\tau$ that selects a value $v_i$ if $v_i \geq T = \mu + \sigma^2/\mu$. Then 
$$
\lim_{n \rightarrow \infty} r^*(\tau) = T =  \mu + \frac{\sigma^2}{\mu}
$$
in case nature is allowed to choose adversarial two-point distributions from $\mathcal{P}_2(\mu,\sigma^2,L)$.
\label{thm:asymptotic_twopoint}
\end{theorem}
\begin{proof}
The proof simply relies on showing that with high probability, as $n \rightarrow \infty$, there will be a value whose weight is at least $T$. 
Note that for any two-point distribution in $\mathcal{P}_2(\mu,\sigma^2,L)$ with positive probability on points $a$ and $b$, it must be the case that $b \geq \mu + \sigma^2/\mu$. Furthermore, the probability mass assigned to the point $b$ is always bounded away from zero (in particular by a constant depending on $\mu, \sigma^2$ and $L$, but independent of $n$). This means that, as $n$ grows large the probability that at least one of the values will be greater or equal than $\mu + \sigma^2/\mu$, approaches $1$. In other words, no matter which adversarial distributions nature chooses, there will also be a constant probability mass assigned to the interval $[\mu + \sigma^2/\mu,L]$ for all of them, and, hence, with high probability one of the values will lie in that interval. Therefore, the probability that the gambler is able to select a value $v_i$ with $v_i \geq T$ approaches $1$.
\end{proof}

We remark that the analysis in the proof of Theorem \ref{thm:asymptotic_twopoint} heavily relies on the fact that $\mu, \sigma^2$ and $L$ are independent of $n$.

\section{Moment bound problems}
	\label{sec:mvs}

{
For the ambiguity set $\mathcal{P}(\mu,d,L)$, consisting of all distributions with mean $\mu$, mean absolute deviation $d$ and support upper bound $L$, we obtain the result in Theorem \ref{thm:mms}. As was already discussed in Section \ref{sec:results}, the obtained thresholds are independent of the support upper bound $L$, and in fact also holds for the ambiguity set $\mathcal{P}(\mu,d)$ in which no support upper bound is assumed. 
In order to illustrate our proof method for this result, we first give an outline of the proof of the thresholds in \cite[Corollary 3.2]{boshuizen1992moment} for the ambiguity set $\mathcal{P}(\mu,\sigma^2,L)$.}
 We again refer the reader to Figure \ref{fig:thresholds} for an illustration of the obtained robust thresholds for this ambiguity set.
% 
%\begin{theorem}[Ambiguity set $\mathcal{P}(\mu,\sigma^2,L)$] Let $n \in \N$,  $\mu, L > 0$ and $\sigma^2 \geq 0$. Let $\mathcal{P}(\mu,\sigma^2,L)$ be the set of all distributions with mean $\mu$, variance $\sigma^2$ and whose support is a subset of the interval $[0,L]$. Let $\mathcal{P}(I_i) = \mathcal{P}(\mu,\sigma^2,L)$ for all $i = 1,\dots,n$. Then the optimal robust threshold in \eqref{eq:opt_robust_tresholds} equals
%\begin{align}
%T(i) & = \mu + \frac{\sigma^2}{\mu}\left[1 - \left(1 - \frac{\mu}{L}\right)^{n-1-i} \right]
%\label{eq:mvs_general}
%\end{align}
%for $i = 1,\dots,n-1$, and $T(n) = 0$. Furthermore, as $n \rightarrow \infty$, the seller's expected robust payoff $r^* = T(0)$  approaches $\mu + \frac{\sigma^2}{\mu}$.
%\label{thm:mvs_general}
%\end{theorem}

This section is structured as follows. We first illustrate the primal-dual method for moment bound problems in Section \ref{sec:prelim_momentbound}, and use it to prove a result of Cox \cite{cox1991bounds} in this language, {which is then used to prove \cite[Corollary 3.2]{boshuizen1992moment}} . We provide the proof of Theorem \ref{thm:mms}, also using the primal-dual moment bound approach, in Section \ref{sec:mmr}. An asymptotic result for the set $\mathcal{P}(\mu,\sigma^2,L)$, similar in spirit as Theorem \ref{thm:asymptotic_twopoint}, is given in Section \ref{sec:asymptotic}.

\subsection{Primal-dual method by example}\label{sec:prelim_momentbound}
For $\mathcal{P}(\mu,\sigma^2,L)$, based on the formulation \eqref{eq:opt_robust_tresholds_v2}, we need to solve the maximin problem
\begin{align}
 & \displaystyle\max_{\Prob\in \mathcal{P}(\mu,\sigma^2,L)}\E_\Prob[\min\{\xi, X\}]
\end{align}
with $\xi=T(i+1)$ considered a constant. The problem can be viewed as a semi-infinite linear optimization problem (LP)
\begin{align}
  \displaystyle\max_{\Prob} & \int_{[0,L]} \min\{t,x\}\, \mathrm{d}\Prob(x)\nonumber  \\
 \text{s.t.} & \int_{[0,L]} 1 \, \mathrm{d}\Prob(x) = 1, \ 
\int_{[0,L]} x \, \mathrm{d}\Prob(x) = \mu, \ 
\int_{[0,L]} x^2 \, \mathrm{d}\Prob(x) = \mu^2 + \sigma^2  \label{coxprimal}
\end{align}
with a finite number of constraints (three in this case), and an infinite number of variables (all distributions $\mathbb{P}$ in the ambiguity set $\mathcal{P}$). The Richter-Rogosinski Theorem (see, e.g.,~\cite{rogosinski1958moments,shapiro2009lectures}) states that there exists an extremal distribution for problem \eqref{coxprimal} with at most three support points. While finding these points in closed form is typically not possible for general semi-infinite problems, we show that this is possible for the problem at hand by resorting to the dual problem; see e.g. \cite{Isii1962} and \cite{popescu2005semidefinite}. To explain this dual approach, we install a small intermezzo, and first cast \eqref{coxprimal} in the more general form of 
generalized moment bound  problems:
\begin{eqnarray}
\text{(P)} \ \  p^*&=\displaystyle\max_{\Prob \in \mathcal{P}(L)} & \E_{}[h(X)]\nonumber  \\
& \text{s.t.} & \E[g_j(X)] = q_j \text{ for } j = 0,\dots,m, \nonumber
\end{eqnarray}
where $X$ is a random variable with nonnegative probability measure $\Prob$, and $h,g_0,\dots,g_m$ are real-valued functions. We use the convention that the function $g_0$ always represents the constraint that $\Prob$ is a probability measure (which is essentially redundant as we already assume that $\mathcal{P}(L)$ is a class of probability distributions). That is, we have $g_0 \equiv 1$, and $q_0 = 1$. The \emph{dual program} of a moment bound problem of the form (P) is given by
\begin{eqnarray}
\text{(D)} \ \    d^*&=\displaystyle\min_{\lambda_0,\dots,\lambda_m} & \sum_{j=0}^m \lambda_j q_j   \nonumber  \\
 & \text{s.t.} & M(x):=\sum_{j=0}^m \lambda_j g_j(x) \geq h(x) \ \ \forall x \in [0,L]. \nonumber
\end{eqnarray}
It is not hard to see that \emph{weak duality} holds, i.e., $p^* \leq d^*$. Furthermore, under mild assumptions, also \emph{strong duality} holds, i.e., $p^* = d^*$. Finally, it is known that the (extremal) distributions solving the primal program (P) are discrete distributions that are supported on at most $n + 1$ points (we refer to \cite{popescu2005semidefinite} for an overview of these claims and further references).

The above primal-dual thinking provides a constructive method for solving a maximin problem by exploiting the dual problem formulation. The solution of the dual problem not only gives an upper bound for the primal problem (by weak duality), but also identifies a candidate worst-case distribution. By testing this candidate distribution in the primal problem, and showing that both problems give the same result (strong duality), one can thus solve the semi-infinite LP. We will now demonstrate this constructive method for $\mathcal{P}(\mu,\sigma^2,L)$ and then for $\mathcal{P}(\mu,d,L)$ in Section \ref{sec:mmr}, both presenting semi-infinite LPs with three constraints. For $\mathcal{P}(\mu,\sigma^2,L)$  this rederives a result in \cite{cox1991bounds}, and for $\mathcal{P}(\mu,d,L)$ this gives a similar result but then for MAD instead of variance. Related results for mean-MAD ambiguity 
in the context of Chebyshev-like inequalities can be found \cite{roos2019chebyshev}.

The dual problem of \eqref{coxprimal} is given by
\begin{eqnarray}
 & \displaystyle\min_{\lambda_0,\lambda_1,\lambda_2} & \lambda_0 + \mu\lambda_1 + (\mu^2 + \sigma^2) \lambda_2  \nonumber  \\
 & \text{s.t.} & \lambda_0+\lambda_1x+\lambda_2x^2 \geq \min\{x,\xi\} \ \ \forall x \in [0,L]. \nonumber
\end{eqnarray}
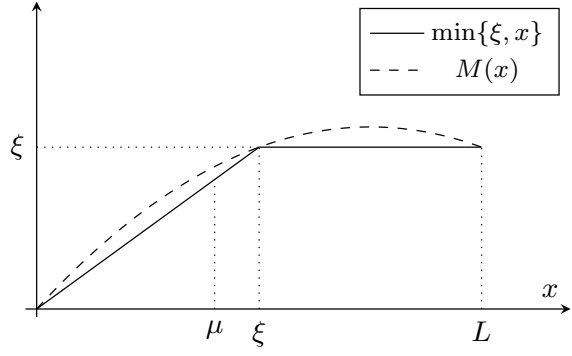
\begin{figure}[h!]
\centering
\begin{tikzpicture}[scale=1.2]
    \begin{axis}[axis lines=middle,xlabel=\footnotesize{$x$},ylabel={},
        domain=0:1,xmin=-0.025,xmax=1.2,ymin=-0.025,ymax=0.95,xtick=\empty,ytick=\empty,
        smooth,samples=500,clip=false,width=3in,height=2in,
        legend entries={$\min\{\xi,x\}$ \\ $M(x)$\\},
        legend style = {font = {\fontsize{8 pt}{12 pt}\selectfont}}]
        \addlegendimage{color=black};
        \addlegendimage{dashed};
        %\draw[dotted] (-2.5,-2) -- (-2.5,2);
        %\draw (-2.5,0) node[below]{-5};
        \draw[dotted] (\mtwo,0) -- (\mtwo,\mtwo);
        \draw (\mtwo,0) node[below]{\footnotesize{$\mu$}};
        \draw[dotted] (\sji,0) -- (\sji,\sji);
        \draw (\sji,0) node[below]{\footnotesize{$\xi$}};
        \draw[dotted] (0,\sji) -- (\sji,\sji);
        \draw (0,\sji) node[left]{\footnotesize{$\xi$}};
        \draw[dotted] (\Ltwo,0) -- (\Ltwo,\sji);
        \draw (\Ltwo,0) node[below]{\footnotesize{$L$}};
        \addplot[color=black,domain=0:\sji] {x};
        \addplot[color=black,domain=\sji:1] {\sji};
        %\addlegendentry{\(\min\{\xi,x\}\)};
        \addplot[color=black,dashed,domain=0:1] {(\Ltwo+\sji)/\Ltwo*x -1/\Ltwo*x^2};
        %\addlegendentry{\(M(x)\)};

 \end{axis}         
\end{tikzpicture}
\caption{Majorant $M(x)$ for the function $\min\{\xi,x\}$ when $\mu -  \frac{\sigma^2 }{L - \mu} \leq \xi \leq \mu +  \frac{\sigma^2}{\mu}.$}\label{fig:firstM}
\end{figure}
The dual problem searches the tightest majorant $M(x)=\lambda_0+\lambda_1x+\lambda_2x^2$, a candidate being the concave parabola that touches $\min\{\xi,x\}$ at $\{0,\xi,L\}$, with
 \begin{equation}
\lambda_0=0, \ \ 
\lambda_1=\frac{L+\xi }{L}, \ \ 
\lambda_2=-\frac{1}{L}
\end{equation}  
and dual objective
\begin{equation}\label{obb4}
\frac{\mu(L + \xi) - (\mu^2 + \sigma^2)}{L}. 
\end{equation} For the primal problem, construct the  distribution $\Prob$ on $\{0,\xi,L\}$ with  probabilities   
\begin{equation}\label{ththth}
\Prob(0)= \frac{L\xi - (L+\xi)\mu + \mu^2 + \sigma^2}{L\xi}, \ \ 
\Prob(\xi)=\frac{L\mu - (\mu^2 + \sigma^2)}{(L-\xi)\xi}, \ \ 
\Prob(L)=\frac{\mu^2 + \sigma^2 - \xi\mu}{(L-\xi)L}.
\end{equation}
This three-point distribution has primal objective value \eqref{obb4}, which thus is the optimum. Note that this is only a distribution when 
\begin{equation}
\mu -  \frac{\sigma^2 }{L - \mu} \leq \xi \leq \mu +  \frac{\sigma^2}{\mu}. 
\end{equation}
For $\xi>\mu +  \frac{\sigma^2}{\mu}$ consider the majorant $M(x)=x$ with dual objective $\mu$. To find a distribution with a matching primal objective, consider all distributions on $[0,\xi]$ in $\mathcal{P}(\mu,\sigma^2, L)$, such as the three-point distribution on $\{0,\mu,\xi\}$. 
 For $\xi<\mu -  \frac{\sigma^2 }{L - \mu}$ consider the majorant $M(x)=\xi$ with dual objective $\xi$. Distributions with a matching primal objective are all distributions in $\mathcal{P}(\mu,\sigma^2, L)$ on $[\xi,L]$, in $\mathcal{P}(\mu,\sigma^2, L)$, including the three-point distribution on $\{\xi,\mu,L\}$. 	Taken together, we obtain the following lemma:

\begin{lemma}[Cox \cite{cox1991bounds}]\label{lem:cox}
Let $\mathcal{P}(\mu,\sigma^2,L)$ be the set of all distribution with mean $\mu$, variance $\sigma^2$ and support contained in the interval $[0,L]$. Then
\begin{equation} \label{eq:cox_solution}
\displaystyle \max_{\Prob \in \mathcal{P}(\mu,\sigma^2,L)}  \int_{[0,L]} \min\{\xi,x\}\, \mathrm{d}\Prob(x) = \left\{
\begin{array}{ll}
\xi & \text{\rm for } \displaystyle 0 \leq \xi \leq \mu -  \frac{\sigma^2 }{L - \mu} \\
\displaystyle \frac{\mu(L + \xi) - (\mu^2 + \sigma^2)}{L} & \text{\rm for }  \displaystyle \mu -  \frac{\sigma^2 }{L - \mu} \leq \xi \leq \mu +  \frac{\sigma^2}{\mu}\\
\mu & \text{\rm for } \displaystyle  \mu +  \frac{\sigma^2}{\mu} \leq \xi \leq L \\
\end{array}
\right..
\end{equation} 
\end{lemma}

For the proof of Theorem \ref{thm:mvs_general} presented below, the max-min result in Lemma~\ref{lem:cox} presents the worst-case scenario that serves as input for finding the robust thresholds. As will become clear, nature chooses as worst-case distributions the three-point distributions in \eqref{ththth}, which is close to a two-point distribution. 
That is, as $\xi \uparrow \mu + \frac{\sigma^2}{\mu}$, the distribution in  \eqref{ththth} converges in probability to the distribution
\begin{align}
\Prob (x) 
= \left\{
\begin{array}{ll}
\displaystyle \frac{\sigma^2}{\mu^2 + \sigma^2} & \text{ for } x = 0 \vspace*{1ex}  \\
\displaystyle \frac{\mu^2}{\mu^2 + \sigma^2} & \text{ for } x = \mu + \frac{\sigma^2}{\mu}
\end{array}
\right..
\label{eq:cox_3to2}
\end{align} 
Note that this is the same distribution as in \eqref{eq:worst_f}.

\subsection{A proof of Theorem \cite[Corollary 3.2]{boshuizen1992moment}} 
\label{sec:mvs_general}
{In this section we give a proof of \cite[Corollary 3.2]{boshuizen1992moment}} by combining Lemma~\ref{lem:cox}, Theorem~\ref{thm:opt_robust_thresholds}, and the next result. 
We will work with the second form of $T(i)$ as given in \eqref{eq:opt_robust_tresholds}, i.e.,
\begin{align}
T(i) &= \mu + T(i+1) - \max_{\Prob_{i+1} \in \mathcal{P}(\mu,\sigma^2,L)} \int_0^{\infty} \min\{T(i+1),x\}\, \mathrm{d} \Prob_{i+1}(x)
\end{align}
for $i = 0,\dots,n-1$, with $T(n) = 0$. Note that we have $\E[X_i] = \mu$ for every $i = 1,\dots,n$.

In the first step of the backwards recursion, we obtain $T(n-1) = \mu$, which follows from the first case of the solution as given in \eqref{eq:cox_solution} by considering $t = 0$. In the next step of the recursion, we therefore start in the second case of \eqref{eq:cox_solution} with $t = \mu$. Remember that the second case was given by the optimal value
$$
p^*(t) = \frac{\mu(L + t) - (\mu^2 + \sigma^2)}{L}
$$
for $\mu -  \sigma^2/(L - \mu) \leq t \leq \mu + \sigma^2/\mu$. 

From Lemma \ref{lem:bound} we know that $\mu = T(n-1) \leq \dots T(1) \leq T(0) \leq \mu + \sigma^2/\mu$.  Since $\mu + \sigma^2/\mu$ is precisely the upper bound of the range of the second case of \eqref{eq:cox_solution}, this means that $T(i)$, for $i = 0,\dots,n-2$, is determined by the recursion
$$
T(i) = \mu + T(i+1) - \frac{\mu(L + T(i+1)) - (\mu^2 + \sigma^2)}{L} = \frac{\mu^2 + \sigma^2}{L} + \left(1 - \frac{\mu}{L}\right)T(i+1),
$$  
i.e., a recursion based on the solution to the second case of \eqref{eq:cox_solution}. Using Proposition \ref{prop:rec} with $\alpha = (\mu^2 + \sigma^2)/L$, $\beta = 1 - \mu/L$ and $\gamma_0 = \mu$, then gives
$$
T(i) = \frac{\mu^2 + \sigma^2}{L}  \left( \sum_{k=0}^{n-2-i} \left(1 - \frac{\mu}{L}\right)^k \right) + \mu\left(1 - \frac{\mu}{L}\right)^{n-1-i}
$$
for $i = 0,\dots,n$. 
Using the identity $\sum_{k=0}^j x^j = (1 - x^{j+1})/(1-x)$, in combination with some calculus, then yields the identity as stated in  \cite[Corollary 3.2]{boshuizen1992moment}. This completes the proof of \cite[Corollary 3.2]{boshuizen1992moment}.

As already illustrated in Section \ref{sec:prelim_momentbound}, the worst-case distribution in every step of the recursion is given by the three-point distribution in  \eqref{ththth} with $\xi = T(i+1)$. 
In Figure \ref{fig:prob_mass} we have given an overview of how the probability mass of this worst-case three-point distribution changes over time. It can be seen that the mass on $L$ approaches $0$.

\begin{figure}[ht!]
\centering
\includegraphics[scale=0.5]{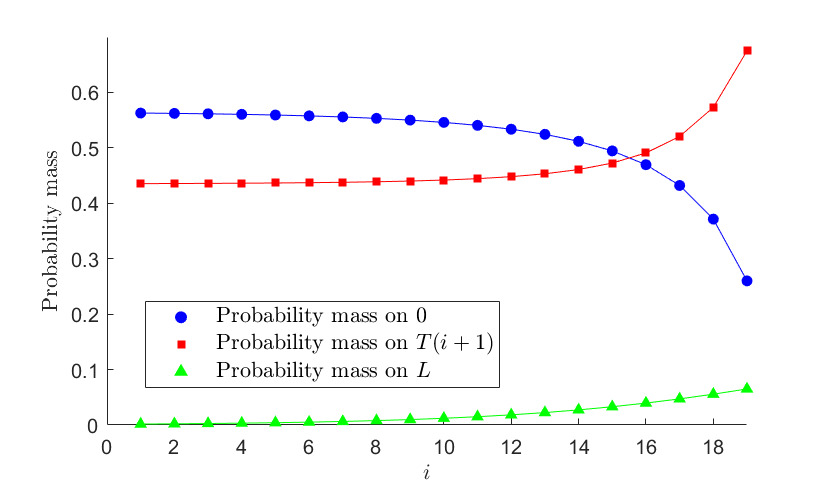}
\caption{Sketch of probability mass on points $\{0,T(i+1),L\}$ for $i = 1,\dots,n-1 = 19$, for worst-case three-point distribution \eqref{ththth} with $\xi = T(i+1)$. The chosen parameters are $\mu = 1$, $\sigma^2 = 1.3\mu$ and $L = 5\mu$.}
\label{fig:prob_mass}
\end{figure}

\subsection{An asymptotic result}
\label{sec:asymptotic}
As for the case of two-point distributions, we can also obtain the payoff of $\mu + \sigma^2/\mu$ approximately (as $n$ grows large) by means of simply setting a threshold of 
\begin{equation}
T = \mu + \frac{\sigma^2}{\mu} - \epsilon
\end{equation}
at every step of the game, for any $\epsilon > 0$. That is, if one is only interested in obtaining the optimal payoff when $n$ is large, and not so much about identifying the optimal strategy at every individual step of the game, then this can be done in a fairly simple manner. Theorem \ref{thm:asymptotic} below is a slightly weaker version of Theorem \ref{thm:asymptotic_twopoint}. We need to choose a slightly lower threshold than in Theorem \ref{thm:asymptotic_twopoint}, as nature can circumvent the reasoning in the proof of Theorem \ref{thm:asymptotic_twopoint} by using three-point distributions.

\begin{theorem}
Let $\epsilon > 0$. Consider the stopping rule $\tau$ that selects the first offer $v_i$ for which $v_i \geq T = \mu + \sigma^2/\mu - \epsilon$. Then 
$$
\lim_{n \rightarrow \infty} r^* \in \left[\mu + \frac{\sigma^2}{\mu} - \epsilon, \mu + \frac{\sigma^2}{\mu}\right]
$$
if every offer is equipped with the ambiguity set $\mathcal{P}(\mu,\sigma^2,L)$.
\label{thm:asymptotic}
\end{theorem}
\begin{proof}
The proof simply relies on showing that with high probability, as $n \rightarrow \infty$, there will be a value whose weight is at least $T$. We do this by showing that, for any $\Prob \in \mathcal{P}(\mu,\sigma^2,L)$, there is always a constant probability mass $q = q(\mu,\sigma^2,L,\epsilon) > 0$ assigned to the interval $[\mu + \sigma^2/\mu-\epsilon,L]$ (we emphasize that this is not true in case $\epsilon = 0$). This means that as $n$ grows large, with high probability there will be at least one offer whose value is at least $\mu + \sigma^2/\mu - \epsilon$, and therefore the seller will always be able to obtain a payoff of at least $\mu + \sigma^2/\mu - \epsilon$ with high probability. We also know from Lemma \ref{lem:bound} that nature can be make sure the seller never gets a robust expected payoff of more than $\mu + \sigma^2/\mu$, so this would then complete the claim.

It remains to show that for any $\Prob \in \mathcal{P}(\mu,\sigma^2,L)$ a constant probability mass $q = q(\mu,\sigma^2,L) > 0$ is assigned to the interval $[\mu + \sigma^2/\mu-\epsilon,L]$. This can be derived directly, e.g., from a result of De Schepper and Heijnen \cite{de1995general}. In particular, from \cite{de1995general} it follows that
$$
\min_{\Prob \in \mathcal{P}(\mu,\sigma^2,L)} \Prob(X \geq \mu + \sigma^2/\mu - \epsilon) = \frac{\epsilon}{L(L - \mu - \sigma^2/\mu)} =: q(\mu,\sigma^2,L).
$$
This completes the proof. 
\end{proof}

We remark that the analysis in the proof of Theorem \ref{thm:asymptotic} heavily relies on the fact that $\mu, \sigma^2$ and $L$ are independent of $n$, otherwise the parameter $q $ might also depend on $n$, i.e., $q = q(\mu(n),\sigma^2(n),L(n))$ which might result in the fact that $\lim_{n \rightarrow \infty} q(\mu(n),\sigma^2(n),L(n)) = 0$.

\subsection{Proof of Theorem~\ref{thm:mms}}
\label{sec:mmr}

In this section we switch attention to the ambiguity set $\mathcal{P}(\mu,d,L)$ with all probability distributions with mean $\mu$, mean absolute deviation $d$, and support contained in the interval $[0,L]$. In order to prove Theorem \ref{thm:mms}, we rely on the following lemma, which can be seen as a mean-MAD-support analogue of the result of Cox \cite{cox1991bounds}. {A related result for mean-MAD ambiguity in the context of Chebyshev-like inequalities is given in \cite{roos2019chebyshev}. 
% where the function within the integral in \eqref{eq: worst-case prob} is $\mathbb{I}_{x \geq \xi}$ instead of $\min\{\xi,x\}$.}

\begin{lemma} Let $\mathcal{P}(\mu,d,L)$ be the set of all distribution with mean $\mu$, MAD $d$ and support contained in the interval $[0,L]$. Then\begin{align} \label{eq: worst-case prob}
    \displaystyle\max_{\Prob \in \mathcal{P}(\mu,d,L)} & \int_0^L \min\{\xi,x\}\, \mathrm{d}\Prob(x) =
        \begin{cases}
     \xi, \quad &  \xi\in[0,\xi_1], \\
     {\mu}-\frac{d(L-\xi)}{2(L-\mu)}, \quad &  \xi\in[\xi_1,\mu], \\
     \xi-\frac{d\xi}{2\mu}, &  \xi\in[\mu,\xi_2], \\
     \mu, &  \xi\in[\xi_2,L],
    \end{cases}
\end{align}
with $\xi_1$ and $\xi_2$ given by
\begin{equation*}
    \xi_1=\mu-\frac{d(L-\mu)}{2(L-\mu)-d},\quad \xi_2=\mu+\frac{d\mu}{2\mu-d}.
\end{equation*}
\label{lem:cox_mudL}
	\end{lemma}
	
\begin{proof}
	The dual of the primal moment bound problem in \eqref{eq: worst-case prob} is given by
\begin{eqnarray}
 & \displaystyle\min_{\lambda_0,\lambda_1,\lambda_2} & \lambda_0 + \mu\lambda_1 + d \lambda_2  \nonumber  \\
 & \text{s.t.} & \lambda_0+\lambda_1x+\lambda_2|\mu-x| \geq \min\{x,\xi\} \ \ \forall x \in [0,L]. \nonumber
\end{eqnarray}
We write $M(x)=\lambda_0+\lambda_1x+\lambda_2|\mu-x|$ and call $M(x)$ the majorant (as in majorizes $\min\{x,\xi\}$). 
The dual problem searches the tightest majorant $M(x)=\lambda_0+\lambda_1x+\lambda_2|\mu-x|$ that minimizes the objective value.  In order to establish \eqref{eq: worst-case prob}, we will construct for all  regimes in \eqref{eq: worst-case prob} a primal-dual pair of feasible solutions with the same objective value, after which strong duality implies optimality.

\begin{figure}[h!]
\centering
\begin{tikzpicture}[scale=1.2]
    \begin{axis}[axis lines=middle,xlabel=\footnotesize{$x$},ylabel={},
        domain=0:1,xmin=-0.025,xmax=1.2,ymin=-0.025,ymax=1.1,xtick=\empty,ytick=\empty,
        smooth,samples=500,clip=false,width=3in,height=2in,
        legend entries={$\min\{\xi,x\}$ \\ $M(x)$\\},
        legend style = {font = {\fontsize{8 pt}{12 pt}\selectfont}},
        legend style={at={(axis cs:0.05,0.7)},anchor=south west}]
        \addlegendimage{color=black};
        \addlegendimage{dashed};
        %\draw[dotted] (-2.5,-2) -- (-2.5,2);
        %\draw (-2.5,0) node[below]{-5};
        \draw[dotted] (\mthree,0) -- (\mthree,\mthree);
        \draw (\mthree,0) node[below]{\footnotesize{$\mu$}};
        \draw[dotted] (\sjithree,0) -- (\sjithree,\sjithree);
        \draw (\sjithree,0) node[below]{\footnotesize{$\xi$}};
        \draw[dotted] (0,\sjithree) -- (\sjithree,\sjithree);
        \draw (0,\sjithree) node[left]{\footnotesize{$\xi$}};
        \draw[dotted] (\Ltwo,0) -- (\Ltwo,\sjithree);
        \draw (\Ltwo,0) node[below]{\footnotesize{$L$}};
        \addplot[color=black,domain=0:\sjithree] {x};
        \addplot[color=black,domain=\sjithree:1] {\sjithree};
        \addplot[color=black,dashed,domain=0:\mthree] {x};
        \addplot[color=black,dashed,domain=\mthree:1] {(1 - (\sjithree-\mthree)/(1 - \mthree))*\mthree + (\sjithree-\mthree)/(1 - \mthree)*x};
 \end{axis}         
\end{tikzpicture}\quad \quad
\begin{tikzpicture}[scale=1.2]
    \begin{axis}[axis lines=middle,xlabel=\footnotesize{$x$},ylabel={},
        domain=0:1,xmin=-0.025,xmax=1.2,ymin=-0.025,ymax=1.1,xtick=\empty,ytick=\empty,
        smooth,samples=500,clip=false,width=3in,height=2in,
        legend entries={$\min\{\xi,x\}$ \\ $M(x)$\\},
        legend style = {font = {\fontsize{8 pt}{12 pt}\selectfont}},
        legend style={at={(axis cs:0.05,0.7)},anchor=south west}]
        \addlegendimage{color=black};
        \addlegendimage{dashed};
        %\draw[dotted] (-2.5,-2) -- (-2.5,2);
        %\draw (-2.5,0) node[below]{-5};
        \draw[dotted] (\mfour,0) -- (\mfour,\sjithree);
        \draw (\mfour,0) node[below]{\footnotesize{$\mu$}};
        \draw[dotted] (\sjithree,0) -- (\sjithree,\sjithree);
        \draw (\sjithree,0) node[below]{\footnotesize{$\xi$}};
        \draw[dotted] (0,\sjithree) -- (\sjithree,\sjithree);
        \draw (0,\sjithree) node[left]{\footnotesize{$\xi$}};
        \draw[dotted] (\Ltwo,0) -- (\Ltwo,\sjithree);
        \draw (\Ltwo,0) node[below]{\footnotesize{$L$}};
        \addplot[color=black,domain=0:\sjithree] {x};
        \addplot[color=black,domain=\sjithree:1] {\sjithree};
        \addplot[color=black,dashed,domain=0:\mfour] {\sjithree/\mfour*x};
        \addplot[color=black,dashed,domain=\mfour:1] {\sjithree};
 \end{axis}         
\end{tikzpicture}
%\caption{Majorants $M(x)$ for the function $\min(\xi, x)$ and $\mu\geq \xi$.}\label{figM1}
\caption{Majorants $M(x)$ for the function $\min(\xi, x)$ for $\mu>\xi$ (left) and $\mu< \xi$ (right).}\label{figM1}
\end{figure}
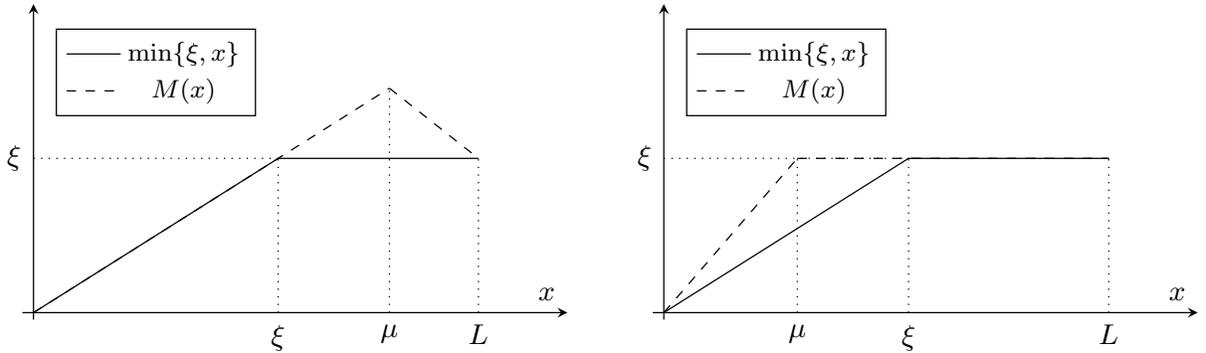
The function $h(x)=\min(\xi,x)$ is piecewise concave, and candidate tightest majorants $M(x)$ are piecewise linear with a possible kink at $x=\mu$.

First, assume $\mu\geq\xi$ and consider the majorant in Figure~\ref{figM1} with 
  \begin{equation}
\lambda_0=-\frac{\mu  \xi -\mu  L}{2 (L-\mu )}, \ \ 
\lambda_1 =-\frac{2 \mu -L-\xi }{2 (L-\mu )}, \ \ 
\lambda_2=-\frac{L-\xi }{2 (L-\mu )}
\end{equation}  
and dual objective
      \begin{equation}\label{obb1}
\mu-\frac{d (L-\xi)}{2 (L-\mu )}. 
\end{equation}  
For the primal problem, we construct a distribution $\Prob$ on $\{0,\xi,L\}$ with  probabilities   
\begin{equation}
\Prob(0) =1-\frac{\mu}{\xi}+\frac{d(L-\xi)}{2 \xi  (L-\mu )}, \ \ 
\Prob(\xi)=\frac{\mu}{\xi}-\frac{d L}{2 \xi  (L-\mu )},\   \
\Prob(L) =\frac{d}{2 (L-\mu )},
\label{eq:mudL_threepoint}{\mu}/{\sigma}
\end{equation}
which is only a distribution when $\xi\geq\xi_1$ with  
\begin{equation}
\xi_1=\mu-\frac{d L}{2(L-\mu)}.
\end{equation}
and when
\begin{equation}
d \leq \frac{2\mu(L-\mu)}{L}.
\end{equation}
The latter condition is always satisfied in order to ensure that $\mathcal{P}(\mu,d,L)$ is nonempty.
The three-point distribution in \eqref{eq:mudL_threepoint} has as objective value the expression in \eqref{obb1}, which thus is the optimum. 

Secondly, let $\xi\leq \xi_1$  and consider the majorant with $\lambda_0=\xi, \lambda_1=0, \lambda_2=0$ and hence the dual solution with objective value $\xi$. For the primal problem, consider \emph{any} distribution supported on $\{\xi,\mu,L\}$. All such distributions give $\xi$ as primal objective.  It is not hard to see that such a distribution always exists.

Thirdly, assume $\mu\leq\xi$ and consider the majorant with $M(0)=0, \, M(\mu)=M(L)=\xi$, see Figure~\ref{figM1}, which results in
\begin{equation*}
    \lambda_0=\frac{\xi}{2} \quad \lambda_1=\frac{\xi}{2\mu}, \quad \lambda_2=-\frac{\xi}{2\mu},
\end{equation*}
and dual objective value
\begin{equation}\label{eq:dual_case3}
    \lambda_0+\lambda_1\mu+\lambda_2d=\xi-\frac{d\xi}{2\mu}.
\end{equation}
Consider here the two-point distribution $\Prob$ with support $\{0,\xi_2\}$ and
\begin{equation*}
\Prob(0) =\frac{d}{2 \mu },\ \ \Prob(\xi_2)=1 - \frac{d}{2 \mu }
\end{equation*}
For any $\mu \leq \xi \leq \xi_2$ the above distribution yields a primal objective of $\xi(1 - d/2\mu)$ which equal the dual objective in \eqref{eq:dual_case3}. 
This leaves  the final case $\xi>\xi_2$. Consider the majorant with $\lambda_1=1$ and hence dual objective $\mu$. One may consider any primal distribution $\Prob$ supported on $\{0,\mu,\xi\}$ (of which there always exists at least one). For such a distribution the primal objective is precisely the expectation of $\Prob$, which is $\mu$. This concludes the proof. 
\end{proof}

	We continue with the proof of Theorem \ref{thm:mms}.
	
	\begin{proof}[Proof of Theorem \ref{thm:mms}.]
		%We roughly follow the proof of Theorem \ref{thm:mvs_general}.
		 We have $T(n) = 0$ as always. Note that if nature always plays the two-point distribution
\begin{align}
\Prob(x) = \left\{\begin{array}{ll}
\displaystyle \frac{d}{2\mu} & \ \ \ \text{for } x = 0 \\
\displaystyle 1 -\frac{d}{2\mu} & \ \ \ \text{for } \displaystyle x = \mu + \frac{d\mu}{2\mu-d}
\end{array}\right.
\label{eq:worst_f_mudL}
\end{align}
		then the seller can never obtain a payoff of more than $\mu + d\mu/(2\mu-d)$. This implies that $\mu \leq T(n-1) \leq T(n-2) \leq \dots \leq T(1) \leq T(0) \leq \mu + d\mu/(2\mu-d)$. Looking at Lemma \ref{lem:cox_mudL}, this means we are always in the third regime for $i = n-1,\dots,0$.

We then have
\begin{align*}
T(i) &= \mu + T(i+1) - \max_{\Prob_{i+1} \in \mathcal{P}(\mu,\sigma^2,L)} \int_0^{\infty} \min\{T(i+1),x\}\, \mathrm{d} \Prob_{i+1}(x) \\
&= \mu + T(i+1) - \left[T(i) - \frac{dT(i)}{2\mu}\right] = \mu + \frac{d}{2\mu}T(i).
\end{align*}
Invoking Proposition \ref{prop:rec} with $\alpha = \gamma_0= \mu$ and $\beta = \frac{d}{2\mu}$ then gives
$$
T(i) = \frac{2\mu^2}{2\mu - d} - \left[\frac{2\mu^2}{2\mu - d} - \mu \right]\left( \frac{d}{2\mu}\right)^{n-1-i},
$$
which converges to $2\mu^2/(2\mu-d)$ as $n \rightarrow \infty$.
\end{proof}

%\red{
%\subsection*{Acknowledgements}
%We are grateful to an anonymous reviewer for pointing out the fact that some of our results (presented as novel in the first version of this work) already appeared in the work of Boshuizen and Hill \cite{boshuizen1992moment}. 
%}
\bibliographystyle{plain}
\bibliography{references}

\appendix

 \section{Regarding correlations}
 \label{app:correlations}
In this section we argue that if correlations are allowed between the random variables $X_1,\dots,X_n$, then the maximin problem in \eqref{eq:minimax_seller} is typically not very interesting. For example, suppose that $I_i = I_j$, so that $\mathcal{P} = \mathcal{P}(I_i) = \mathcal{P}(I_j)$,  for all $i,j = 1,\dots,n$ (as we do for the results in Section \ref{sec:mvs_twopoint}).

If nature chooses a common distribution $\Prob$ from $\mathcal{P}$ for all offers, whose mean $\mu$ is as small as possible, and fully correlates the random variables $X_1,\dots,X_n$ (meaning that all variables always have their realized value in common), then it is clear that in expectation the seller can obtain a payoff of at most $\mu$ regardless of its chosen stopping rule. On the other hand, always simply selecting the first offer will definitely guarantee a payoff of $\mu$ independent of which (possibly correlated) distributions nature chooses. This then yields the solution to \eqref{eq:minimax_seller} with $r^* = \mu$.

We remark that more might be possible in case one, in some way, restricts the amount of correlation that is allowed between the random variables. In the context of prophet inequalities, we refer the interested reader to the work of Immorlica et al.~\cite{immorlica2020prophet}, and references therein, for possible models along this line.

\section{Non-emptyness conditions for mean-dispersion ambiguity sets}\label{app:non-empty}
The following lemma provides a formal proof of the claims made in Section \ref{sec:prelim} regarding the non-emptiness of the ambiguity sets $\mathcal{P}(\mu,\sigma^2,L)$ and $\mathcal{P}(\mu,d,L)$.

\begin{lemma}
Let $\mu,\, \sigma,\, d,\, L \geq 0$. The ambiguity set $\mathcal{P}(\mu,\sigma^2,L)$, consisting of all distributions with mean $\mu$, variance $\sigma^2$ and support contained in $[0,L]$, is nonempty if and only if
$$
\sigma^2 \leq \mu(L-\mu).
$$
The set $\mathcal{P}(\mu,d,L)$, consisting of all distributions with mean $\mu$, mean absolute deviation $d$ and support contained in $[0,L]$, is non-empty if and only if
$$
d \leq \frac{2\mu(L-\mu)}{L}.
$$
\label{lem:nonemptiness}
\end{lemma}
\begin{proof}
We first consider the mean-variance-support ambiguity set $\mathcal{P}(\mu,\sigma^2,L)$. Consider the two-point distribution
\begin{align}
\Prob (x) 
= \left\{
\begin{array}{ll}
\displaystyle \frac{\sigma^2}{\mu^2 + \sigma^2} & \text{ for } x = 0 \vspace*{1ex}  \\
\displaystyle \frac{\mu^2}{\mu^2 + \sigma^2} & \displaystyle \text{ for } x = \mu + \frac{\sigma^2}{\mu}
\end{array}
\right..
\end{align} 
First suppose that $\sigma^2 \leq \mu(L-\mu)$. This condition is equivalent to $\mu + \sigma^2/\mu \leq L$, which means that $\Prob \in \mathcal{P}(\mu,\sigma^2,L)$, and, hence, $\mathcal{P}(\mu,\sigma^2,L)$ is non-empty. For the converse we will use a standard result of which we next include a proof for completeness. Suppose that $\mathcal{P}(\mu,\sigma^2,L)$ is non-empty and consider any random variable $X \sim \Prob'$ where $\Prob' \in \mathcal{P}(\mu,\sigma^2,L)$. If $L = 1$, we have
$$
\sigma^2 = \E[(X-\mu)^2] = \E[X^2] - \mu^2 \leq \E[X] - \mu^2 = \mu(1-\mu) = \mu(L-\mu)
$$
where in the inequality we use the fact that $X$ is supported on $[0,1]$. For general values of $L$, the claim follows by considering the rescaled random variable $Y = X/L$.

We continue with the mean-MAD-support  ambiguity set $\mathcal{P}(\mu,d,L)$. Consider the two-point distribution
\begin{align}
\Prob (x) 
= \left\{
\begin{array}{ll}
\displaystyle \frac{d}{2\mu} & \text{ for } x = 0 \vspace*{1ex}  \\
\displaystyle 1 - \frac{d}{2\mu} & \displaystyle \text{ for } x = \frac{2\mu^2}{2\mu-d}
\end{array}
\right..
\label{eq:non-empty_MAD}
\end{align} 
First suppose that $d \leq 2\mu(L-\mu)/L$. This condition is equivalent to $2\mu^2/(2\mu-d) \leq L$ and so $\Prob \in \mathcal{P}(\mu,d,L)$. Conversely, suppose that $\mathcal{P}(\mu,d,L)$ is non-empty and let $\Prob' \in \mathcal{P}(\mu,d,L)$. If we put all the probability mass in $[0,\mu)$ on its conditional mean in that interval, and, similarly, all the probability mass in $[\mu,L]$ on its conditional mean in that interval, then the resulting two-point distribution $\Prob''$, supported on, say, $\{x,y\}$ with $x \leq y$, still has mean $\mu$ and MAD $d$ (because of the linear nature of the MAD), and is therefore contained in $\mathcal{P}(\mu,d,L)$. It is not hard to see that the distribution $\Prob$ in \eqref{eq:non-empty_MAD} is the two-point distribution with mean $\mu$ and MAD $d$ for which both points in the support are chosen as small as possible (among all choices for two-point distributions with mean $\mu$ and MAD $d$). This means that 
$$
\frac{2\mu^2}{2\mu-d} \leq y \leq L
$$
where the second inequality holds as $\Prob'' \in \mathcal{P}(\mu,d,L)$. Again, the resulting inequality is equivalent to $d \leq 2\mu(L-\mu)/L$, which completes the proof.
\end{proof}

\end{document}